\documentclass[a4paper,11pt,fleqn]{article}
\pdfoutput=1  

\usepackage[utf8]{inputenc}
\usepackage[T1]{fontenc}
\usepackage{lmodern}
\normalfont \DeclareFontShape{T1}{lmr}{bx}{sc} { <-> ssub * cmr/bx/sc }{}
\usepackage{microtype}
\usepackage{mathtools}
\usepackage[margin=2.6cm]{geometry}
\usepackage{cite}
\usepackage{enumerate}
\usepackage{amsmath,amssymb}
\usepackage{authblk}
\usepackage{xcolor}
\usepackage{dsfont}
\usepackage{courier}
\usepackage{enumitem}
\usepackage{xspace}
\usepackage{tikz}
\usepackage[amsmath,thmmarks,amsthm]{ntheorem}
\usepackage[font=footnotesize,labelsep=period,labelfont=bf]{caption}
\usepackage{subcaption}
\usepackage{booktabs}

\setlist{itemsep=0em}

\definecolor{ETHblue}{HTML}{1F407A}
\usepackage[colorlinks,allcolors=ETHblue]{hyperref}
\usepackage[nameinlink,capitalize]{cleveref}

\usetikzlibrary{decorations.pathreplacing}
\usetikzlibrary{arrows}

\theoremstyle{plain}
\theoremseparator{.}
\theorembodyfont{\normalfont\upshape}
\newtheorem{conjecture}{Conjecture}
\newtheorem{theorem}{Theorem}
\newtheorem{definition}{Definition}
\newtheorem{observation}{Observation}
\newtheorem{lemma}{Lemma}
\newtheorem{fact}{Fact}
\newtheorem{corollary}{Corollary}
\theoremstyle{remark}
\newtheorem{claim}{Claim}
\newenvironment{claimproof}[1][\proofname]
  {\begin{proof}[#1]}
  {\hfill$\Diamond$\end{proof}}

\crefname{figure}{Figure}{Figures}
\crefname{lemma}{Lemma}{Lemmata}
\crefname{theorem}{Theorem}{Theorems}
\crefname{conjecture}{Conjecture}{Conjectures}
\crefname{claim}{Claim}{Claims}
\crefname{fact}{Fact}{Facts}
\crefname{corollary}{Corollary}{Corollaries}
\crefname{section}{Section}{Sections}
\crefname{subsection}{Subsection}{Subsections}
\crefname{enumi}{}{}
\crefname{equation}{}{}
\DeclarePairedDelimiter{\abs}{\lvert}{\rvert}

\renewcommand{\epsilon}{\varepsilon}
\newcommand{\eps}{\epsilon}
\newcommand{\alg}{\ensuremath{\textup{\textsc{alg}}}\xspace}
\newcommand{\Opt}{\textup{\textsc{opt}}\xspace}
\newcommand{\algOff}{\ensuremath{\textup{\textsc{off}}}\xspace}
\renewcommand{\algOff}{\ensuremath{\textup{\textsc{opt}}}\xspace}
\newcommand{\sequen}{\textnormal{\textup{\textsc{Robin}}}\xspace}
\newcommand{\sequenbold}{\textnormal{\textbf{\textsc{Robin}}}}

\newcommand{\R}{\mathds{R}}

\newcommand{\Block}{B\xspace}
\newcommand{\block}{b\xspace}
\newcommand{\peri}{f\xspace}
\newcommand{\cent}{h\xspace}
\newcommand{\choice}{c\xspace}
\newcommand{\group}{g\xspace}
\newcommand{\Vertex}{N\xspace}
\newcommand{\vertex}{n\xspace}
\newcommand{\phase}{P\xspace}
\newcommand{\phaselength}{m\xspace}
\newcommand{\m}{s\xspace}
\renewcommand{\Delta}{A}
\newcommand{\E}{\mathds{E}\xspace}
\renewcommand{\P}{\text{Prob}\xspace}
\newcommand{\N}{\ensuremath{\mathds{N}}\xspace}
\newcommand{\sphere}{\ensuremath{\mathds{S}}\xspace}
\newcommand{\algR}{\ensuremath{\textup{\textsc{rand}}}\xspace}
\newcommand{\algDC}{\ensuremath{\textup{\textsc{dc}}}\xspace}
\newcommand{\algDCT}{\ensuremath{\textup{\textsc{dc-tree}}}\xspace}
\newcommand{\algWFA}{\ensuremath{\textup{\textsc{wfa}}}\xspace}
\newcommand{\cost}[1]{\ensuremath{\textup{\text{cost}}(#1)}\xspace}

\begin{document}

\title{\textbf{Time-Optimal \textit{k}-Server}}

\author[1]{Fabian Frei}
\author[2]{Dennis Komm}
\author[2]{Moritz Stocker}
\author[2]{Philip Whittington}

\affil[1]{CISPA Helmholtz Institute for Information Security \authorcr {\small\texttt{fabian.frei@cispa.de}}}
\affil[2]{Department of Computer Science, ETH Zurich \authorcr {\small\texttt{\{dennis.komm, moritz.stocker, philip.whittington\}@inf.ethz.ch}}}
\date{}

\maketitle

\begin{abstract}
  \noindent The time-optimal $k$-server problem minimizes the time spent
  serving all requests instead of the distances traveled.  We give a lower
  bound of $2k-1$ on the competitive ratio of any deterministic online
  algorithm for this problem, which coincides with the best known upper bound
  on the competitive ratio achieved by the work-function algorithm
  for the classical $k$-server problem. 
  We provide further lower bounds of $k+1$ for all Euclidean spaces and $k$ for
  uniform metric spaces. For the latter, we give a matching $k$-competitive
  deterministic algorithm.  Our most technical result, proven by applying Yao's
  principle to a suitable instance distribution on a specifically constructed
  metric space, is a lower bound of $k+\mathcal{O}(\log k)$ that holds even for
  randomized algorithms, which contrasts with the best known lower bound for
  the classical problem that remains polylogarithmic. 

  With this paper, we hope to initiate a further study of this natural yet neglected problem.
\end{abstract}

\section{Introduction}

The \emph{$k$-server problem}, introduced by Manasse et al.\ in
1988~\cite{MMS1988}, has been repeatedly referred to as ``the holy grail'' of
online computation~\cite{BN2009,BBN2010,CK2021}.  In particular the $k$-server
conjectures about the best competitive ratios achievable by deterministic and
randomized algorithms have inspired intensive research for many decades by now. 

The $k$-server problem can be considered for any natural $k$ and any given
nonempty metric space $\mathcal M=(M,d)$ together with an \emph{initial
configuration} $C_0\in M^k$.  An instance is a sequence $r_1,\dots,r_n\in M$ of
point requests, revealed one by one.  An online algorithm answers each request
$r_j$, knowing only the already revealed requests, with a configuration $C_j\in
M^k$ such that $C_j(i)=r_j$ for some $i$, which we describe as
\emph{server $s_i$ serving request $r_j$}.  Hence, any solution is a sequence of
configurations $C_1,\dots,C_n\in M^k$ describing the movements of $k$ servers
such that each requested point is served by moving a server there.  The next
request is revealed only once the previous has been served.  The cost of a
solution is traditionally defined as the total of all distances traveled by all
servers, yielding the classical, well-researched $k$-server problem.  Another
possibility is to count towards the cost not all but only the \emph{maximum}
distance traveled by any server between any two configurations.

We argue that this alternative cost definition yields a problem that is just as
natural and interesting as the variant usually considered.  We refer to the two
ways of defining the cost as the distance model and the time model:  

\begin{description}
  \item[The distance model.] 
  This is the well-known classical variant. The cost incurred by the algorithm
  to serve request $r_j$ is the \emph{sum} of all distances traveled by the
  servers when changing from configuration $C_{j-1}$ to configuration $C_j$.
  The cost of a solution $C_1,\dots,C_n$ is thus $\sum_{j=1}^n\sum_{i=1}^k
  (C_{j-1}(i)-C_j(i))$, the total of all distances moved by all servers. 
  We can imagine that we need to pay some fuel cost for all server
  movements.\footnote{We remark that moving more than one server per time step
  cannot save cost in this model as the algorithm incurs the same cost for a
  movement, no matter when it happens.  Nevertheless, moving multiple servers
  per time step can be helpful to simplify the description of an algorithm. 
  Moreover, it is easy in the distance model to convert an algorithm that moves
  some servers simultaneously into a \emph{lazy} algorithm, which does not do
  this, but incurs at most the same cost~\cite{Kom2016}.} \item[The time
  model.] The cost for serving a request is the \emph{maximum} distance
  traveled by the servers; that is, the total cost is $\sum_{j=1}^n\max_{i=1}^k
  (C_{j-1}(i)-C_j(i))$. 
  Hence, an optimal solution minimizes the total waiting time incurred by the
  requests until they are served.  We thus refer to the problem in this model
  as the \emph{time-optimal $k$-server problem}. 
\end{description}

In many situations where requests need to be served, it is paramount to react
as fast as possible rather than moving less. We might consider ambulances or
police cars being called to the scene of an emergency, but a perfect 
example was ironically given by the famous seminal paper~\cite{MMS1988} that 
used it to introduce the classical $k$-server problem instead: 
planning the motions of a hard disk with $k$ heads. 
Rather than caring about minimizing hidden head movements, the typical user
wants optimized reading and writing speeds. 

Despite this, past research has focused almost exclusively on the distance
model; see \cref{subsec:related}.

We hope to initiate a new line of research that gives the time-optimal
$k$-server problem the attention it deserves.  As usual, the main goal is to
determine the best competitive ratios achievable by deterministic and
randomized algorithms.  We prove several lower bounds that provide an
exponential improvement over what has been known until now, and design a
deterministic algorithm matching these bounds on uniform metric spaces. 

Despite this, a significant gap remains between the upper and lower bounds on
general metric spaces.  Closing this gap seems to be as interesting as
answering the open questions for the distance model. 
Indeed, we expect a fruitful, mutually informative interaction between the two
models.

\subsection{Related Work}\label{subsec:related}

In this section, we review known results and open questions for the $k$-server
problem in both models.  For the classical distance model, to which countless
papers have been dedicated
over the last four decades, we restrict ourselves to the most important
milestones; and at the same time introduce some basic notions and concepts that
we make use of later.  For the time model, however, we can give a full account. 

\paragraph*{The Distance Model.}\label{subsubsec:relateddistance}
We note that on metric spaces with at most $k$ points the algorithm keeping one
server on each point is trivially 1-competitive.  On a clique with $k+1$ or
more vertices, there is a very simple lower bound of $k$: For any algorithm,
there is an instance that always asks for an unoccupied vertex, causing a cost
of $1$ with each request, whereas the offline strategy of serving an unoccupied
vertex with a server currently at a point that is not among the next $k-1$
requests spends at most a cost of $1$ for every $k$ requests.  
(The $k$-server problem on cliques of size $N$ is also equivalent to the paging
problem with $N$ pages and a cache of size $k$~\cite{BE1998,Kom2016}.)

The lower bound of $k$ can be generalized to all nontrivial metric spaces, that
is, \emph{any} metric space with at least $k+1$ points.  The proof is not too
complicated, using a nonconstructive argument of existence for the optimal
algorithm competing with the given online
algorithm~\cite[Thm.~10.1,~Thm.~4.4]{BE1998,Kom2016}.

Manasse et al.\ already conjectured that there is a matching upper bound for
$k$-server when they introduced the problem in 1988~\cite[Conjecture~4]{MMS1988}.

\begin{conjecture}[{{\textit{k}-Server Conjecture}}]
  For any metric space and any $k\ge 1$, there is a $k$-competitive algorithm
  for the $k$-server problem in the distance model. 
\end{conjecture}

They also proved the conjecture in the two very special cases of only $k=2$
servers and of $k$ servers on metric spaces with only $k+1$
points~\cite[Thms.~5,~6]{MMS1988}.  Note that the journal version of this
seminal paper appeared two years later under a different title~\cite{MMS1990}.
This is also how long it took for anyone to establish any upper bound on the
competitive ratio that did not grow with the input length $n$ but only with
$k$, albeit with an exponential dependence~\cite[Thm.~2]{FiatRR90,FiatRR94}.
Another four years later, Koutsoupias and Papadimitriou finally managed to
shatter the exponential barrier by analyzing the so-called \emph{work-function
algorithm} (independently proposed as a candidate by multiple
researchers~\cite[Sect.~2.3]{CL1991b}), which dramatically lowered the upper
bound to $2k-1$~\cite[Thm.~4.3]{KP1995a}.
Emek et al.~\cite{EFKR2010} have shown that \algWFA is also \emph{strictly}
$(4k-2)$-competitive. This represents the best known upper bound on the rarely
considered strict competitive ratio. Note that lower bounds on the competitive
ratio are also lower bounds on the strict competitive ratio. To the best of our
knowledge, no better lower bounds specific to the strict competitive ratio have
been proposed. 

To date, \algWFA remains to be the algorithm with the best known
competitive ratio for general metric spaces. We are thus left with a constant
factor of essentially $2$ (and $4$ for the strict competitive ratio) between
the upper and lower bound, despite all the efforts by the research community
trying to improve the upper bound or disprove the \textit{k}-server conjecture
for multiple decades.  It is conjectured that \algWFA even attains the coveted
competitive ratio of $k$.

But at least for some very special metric spaces it was possible to prove the
\textit{k}-server conjecture.  A prominent example is the real line.  Here,
Chrobak et al.~\cite{CKPV1991} were able to give the matching upper bound of
$k$ on the competitive ratio by introducing and analyzing the so-called
\emph{double coverage} algorithm (\algDC) with the following, somewhat
unintuitive behavior.  When a point $x$ is requested, \algDC moves not only
one, but two servers towards $x$: a closest server to the right of $x$ and a
closest server to the left of $x$. (If there is a tie between multiple servers
on the same location, one of them is chosen arbitrarily.) They move at the same
speed and both stop once one of them has reached $x$.  (This algorithm clearly
works not only on the real line, but also on the rational and integer line, and
on any path graph.) There is a straightforward generalization of \algDC to
trees: All servers synchronously move toward the request until it is served,
stopping prematurely only if another server closer to the target appears on
their path.  This algorithm, called \algDCT, was described and proven to be
$k$-competitive by Chrobak and Larmore~\cite{CL1991}.  Lastly, Koutsoupias and
Papadimitriou~\cite{KP1996} showed that the $k$-server conjecture also holds on
any metric space with $k+2$ points.

For randomized algorithms, all the proofs for a lower bound of $k$ from the
deterministic case do not work anymore.  Instead, a lower bound of $\Omega(\log
k)$ has been known for a long time:  Let $H_k\coloneqq\sum_{j=1}^{k}(1/j)$ denote the
$k$-th harmonic number.  On a clique of $k+1$ points, there is no deterministic
algorithm better than $H_k$-competitive (in expectation) on the instance
distribution which requests any point except the just requested one with the
same probability $1/k$. Yao's principle~\cite{Yao1977} now transforms this
lower bound for deterministic algorithms on an input distribution into a lower
bound for randomized algorithms on fixed
instances~\cite[Thm.~2.11]{BE1998,Kom2016}.  This has remained the best lower
bound for multiple decades.  (A series of
results~\cite{BKRS1992,BLMN2003,BBM2006} showed a lower bound that is only
$\Omega({\log k}/({\log\log k}))$ instead of $\Omega(\log k)$ but in exchange
works on any metric space of at least $k+1$ points.) This lead to the following
folklore conjecture~\cite[Conj.~2]{Kou2009}.

\begin{conjecture}[{{Randomized \textit{k}-Server Conjecture}}]
  For any $k\ge 1$, there is a $\mathcal{O}(\log k)$-competitive randomized
  algorithm for the $k$-server problem in the distance model.
\end{conjecture}

This conjecture was refuted in 2023 by Bubeck et al.~\cite{BCR2023}, who were
able to slightly lift the lower bound to $\Omega((\log k)^2)$ for one special
metric space.  Note, however, the remaining exponential gap between this
improved lower bound and the best known upper bound for randomized algorithms,
which is still $2k-1$ and attained by the deterministic \algWFA. 

A remarkable positive contribution by Bansal et al.~\cite{BBMN2011} is a
randomized algorithm that is, on metric spaces with $N$ points,
${\mathcal{O}}((\log k)^2(\log N)^3\log\log N)$-competitive in expectation, and
thus even with high probability by a general result by Komm et
al.~\cite{KKKM2022}.  Their algorithm thus obtains a competitive ratio that is
polylogarithmic in $k$ whenever $N$ is polynomial in $k$.  A more detailed
discussion of the results up until 2009 is found in a survey by
Koutsoupias~\cite{Kou2009}.

The most relevant results for the \textit{k}-server problem in the distance
model on general metric spaces can thus be summarized as follows (see also
\cref{tab:results}): The best known upper bound on the competitive ratio is
$2k-1$ (for both deterministic and randomized algorithms) and attained by the
work-function algorithm \algWFA~\cite{KP1995a}.  The best known lower bound on
the competitive ratio is $\Omega((\log k)^2)$ for randomized
algorithms~\cite{BCR2023}, and $k$ for deterministic algorithms~\cite{MMS1988}.
A matching upper bound of $k$ is attained on the line by the double coverage
algorithm \algDC~\cite{CKPV1991}, and on trees by the generalized variant
\algDCT~\cite{CL1991}.

\begin{table}
\centering
  \caption{Overview of the best known upper and lower bounds (in the upper and
  lower row, respectively) on the competitive ratio for the $k$-server problem
  in the distance and time model, for deterministic algorithms on the line
  and general metric spaces, and additionally for randomized algorithms on general metric spaces.}
  \label{tab:results}
  \begin{tabular}{r@{} c c c c c c }
  \toprule
  &\multicolumn{3}{c}{Distance model}&\multicolumn{3}{c}{Time model}\\
  \cmidrule(r){2-4} \cmidrule(lr){5-7}
    & Line & General & Randomized & Line & General & Randomized \\
  \cmidrule(r){2-2} \cmidrule(lr){3-3} \cmidrule(lr){4-4} 
  \cmidrule(lr){5-5} \cmidrule(lr){6-6}
  \cmidrule(lr){7-7}
    & $k$ \cite{CKPV1991} & $2k-1$ \cite{KP1995a} & $2k-1$ \cite{KP1995a} &
      $(k^2+k)/2$\cite{KT2004} & $2k^2-k$ \cite{KP1995a, KT2004}  & $2k^2-k$ \cite{KP1995a, KT2004}  \\
    & $k$ \cite{MMS1988} & $k$ \cite{MMS1988} & $\Omega((\log k)^2)$ \cite{BCR2023} & $k+1$ [Th.~\ref{thm:kplusoneline}] & $2k-1$ [Th.~\ref{thm:lowertwokminusonefinite}] & $k+H_k-1$ [Th.~\ref{thm:lowertwokminusonerandomized}] \\
    \bottomrule
  \end{tabular}
\end{table}

\paragraph*{The Time Model.}\label{subsubsec:relatetime}
There is almost no research published on the time model. This is astonishing,
seeing its nontriviality and natural motivation, but allows us to describe
here all existing results. 

Clearly, changing from the distance model to the time model helps the algorithm
because it can now move servers synchronously at no additional cost. However,
note that this does not at all imply an improved competitive ratio. Since the
advantages of the time model can be utilized by both the online algorithm and
solution it competes with, the ratio could a priori improve, stay the same, or
get worse, and in different ways for different metric spaces. 

A first result distinguishing the two variants was given by Koutsoupias and
Taylor \cite{KT2004}.  
As they remark~\cite[Sect.~5]{KT2004}, any upper bound for the distance model
implies a $k$~times larger upper bound for the time model: On the one hand, any
algorithm for the distance model also works for the time model with the same or
even lower cost.  On the other hand, the optimal algorithm against which the
online algorithm competes saves at most a factor of $k$ in the time model
compared to its cost in the distance model.  Overall, the competitive ratio
increases by at most a factor of $k$ when an algorithm for the distance model
is used in the time model.  The $(2k-1)$-competitive \algWFA thus implies an
upper bound of $2k^2-k$.  For trees it immediately follows that \algDCT, which
is $k$-competitive in the distance model, is $k^2$-competitive in the time
model.  Koutsoupias and Taylor~\cite[Thm.~3]{KT2004} lowered this bound to
$k(k+1)/2$ with a straightforward adjustment of the potential function used in
the analysis of \algDCT.  Note that both bounds are in $\mathcal{O}(k^2)$.
They concluded with a lower bound of $3$ for $k=2$ servers on the real
line~\cite[Thm.~4]{KT2004}, but provided no results for~$k\geq 3$.

No further results on the time model have been published to the best of our
knowledge.  We believe that the reason for this variant receiving far less
attention is at least in part that it is less straightforward to analyze since
we can no longer restrict ourselves to lazy algorithms,  which move only one
server at a time.  Indeed, we do not see any reason why the time model should
be considered less interesting or realistic than the distance model.  Using
synchronous movements of the servers to respond more quickly to future requests
is a very natural and reasonable strategy. 

\subsection{Contribution}

The goal of this paper is to popularize the time-optimal $k$-server problem and
establish it as a research object that is well worth the community's attention.
Not only do we believe the $k$-server problem to be as natural and well
motivated in the time model as in the distance model, we also expect the
insights that stand to be gained by analyzing the time model to shed further
light on the long-standing open questions about the classical problem. 

Our main technical contribution are lower bounds for the time-optimal
$k$-server problem that are at least as large as and often larger than the best
known lower bounds for its classical twin.

For deterministic algorithms, we achieve the lower bound of $k$ on all
unweighted graphs with at least $k+1$ vertices (which excludes only cases where
the competitive ratio is trivially $1$).  We then go above this, albeit just
barely, with a lower bound of $k+1$ for a large class of metric spaces that
includes all Euclidean spaces, infinite grids, and large cycles. More
intriguingly still, we construct a specific metric space (super-exponentially
large in $k$ but still finite and of diameter only $3$) for which we can prove
a lower bound of $2k-1$, which is exactly the best known upper bound for the
classical $k$-server problem. 
We additionally prove an intermediate lower bound of $3k/2$ on a small graph
with only $\mathcal{O}(k)$ vertices. 

Our most interesting and challenging result is a lower bound on the expected
competitive ratio for randomized algorithms for the time-optimal $k$-server
problem.  The bound is $k+H_k-1$, and thus exponentially larger than the
recently proven polylogarithmic bound in the distance model.

To round off the results, we provide an algorithm that exactly matches our
first lower bound on uniform metric spaces. 

We summarize the results again in the order that we prove them.  We provide for
the time-optimal $k$-server problem a deterministic algorithm that is
\textbf{(1a)}\label{algo:1} $\Delta k$-competitive on graphs with aspect ratio
$\Delta$ (\cref{corollary:uniformD}) and \textbf{(1b)}\label{algo:2} strongly
$k$-competitive on uniform metric spaces (\cref{corollary:uniform}); then prove for deterministic
algorithms lower bounds of \textbf{(2a)} $k$ on every graph with more than $k$
vertices (\cref{thm:lowerboundk}), \textbf{(2b)} $k+1$ on any Euclidean space,
infinite grid (including the real line and the integer line), and cycles on at
least $2k+6$ points (\cref{thm:kplusoneline}), \textbf{(2c)} $3k/2$ on a small
graph with fewer than $8k$ vertices (\cref{thm:threehalfk}), and \textbf{(2d)}
$2k-1$ on a large graph of diameter $3$ (\cref{thm:lowertwokminusonefinite});
and prove for randomized algorithms a lower bound on the expected competitive
ratio of \textbf{(3)} $k+H_k-1$ on a large graph of diameter $3$
(\cref{thm:lowertwokminusonerandomized}).  For the strict
competitive ratio we have additionally obtained lower bounds of $5k/4$ and
$5k/6$ for deterministic (\cref{thm:strictlowerline}) and randomized algorithms
(\cref{thm:strictlowerlinerandomized}), respectively, on the real line.

Despite the significantly raised lower bounds, a gap still remains with a
linear factor between the upper and lower bounds. 

\section{Preliminaries}

We introduce the basic notions necessary to state our results.  For any
nonnegative integer $n$, we use the notation $[n]\coloneqq\{1,\dots,n\}$. As mentioned
before, $H_k\coloneqq\sum_{j=1}^k (1/j)$ denotes the $k$-th harmonic number.

\paragraph*{Online Algorithms and Competitive Analysis.}\label{sec:defonline}
Online problems are modeled as \emph{request-answer games} where
an online player \alg (the algorithm) plays against an adversary that designs the 
instance piece by piece while reacting to \alg's actions.
Formally, an instance $I$ consists of $n$ \emph{requests}, $I=(x_1,\dots,x_n)$,
which are presented to \alg in $n$ consecutive time steps; 
$n$ is by default not known to \alg a priori.
Whenever an $x_i$ is presented as a request,
\alg has to provide an \emph{answer} $y_i$ to $x_i$, where $y_i$ depends only
on the prefix $x_1,\dots,x_i$ of already presented requests.  The full answer
sequence $\alg(I)=(y_1,\dots,y_n)$ is the solution computed by \alg on $I$.
Any solution $\alg(I)$ to $I$ is assigned a \emph{cost}, usually a positive
real number denoted by $\cost{I,\alg(I)}$. (We assume here a minimization
problem; for maximization problems \emph{cost} is often called \emph{profit}.)
An \emph{optimal solution} for $I$, denoted by $\Opt(I)$, is a solution with
optimal cost $\cost{I,\Opt(I)}$, i.e., minimal cost across all solutions to
$I$ in our case.  Note that $\Opt(I)$ and $\cost{I,\Opt(I)}$ can
generally be computed only once $I$ is known.

\emph{Competitive analysis}, introduced by Sleator and Tarjan in
1985~\cite{ST1985}, measures the cost of the solution computed by \alg relative
to an optimal solution computed by an offline algorithm \Opt.  Formally, \alg
is called \emph{$c$-competitive} if there is a constant $\alpha$ so that, for
every instance $I$, it holds that $\cost{I,\alg(I)} \le c\cdot \cost{I,\Opt(I)}
+ \alpha$.  If this inequality holds even with $\alpha=0$ for all instances,
\alg is called \emph{strictly} $c$-competitive. 

Similarly, a randomized online 
algorithm \algR is said to be \emph{$c$-competitive in expectation} or have an
\emph{expected competitive ratio of $c$} if there is a constant $\alpha$ so
that, for every instance $I$, it holds that $\E\left[\cost{I,\algR(I)}\right]
\le c\cdot \cost{I,\Opt(I)} + \alpha$.
If the inequality holds with $\alpha=0$, then \algR is said to be
\emph{strictly $c$-competitive in expectation} or have a
\emph{strict expected competitive ratio of $c$}.

\paragraph*{\boldmath The $k$-Server Problem.}\label{sec:kserver} 
The intuition behind the $k$-server problem is best described by a scenario
where we are given a metric space and $k$ \emph{servers}, each occupying one
point in the space.  A request $x$ is any point of the space, and an answer
must specify at least one server that is moved to $x$ in response.
However, it is possible to move other servers as well.
A request is revealed once the previous one has been served.  The goal is to
minimize either the total distance traveled after answering all requests (the
distance model) or the overall time spent (the time model). 

Formally, let $\mathcal{M}=(M,d)$ be any metric space. (In particular, the
distance function $d\colon M^2\to\R$ is zero on the pairs $(x,x)$ for $x\in M$,
otherwise positive, symmetric, and satisfies the triangle inequality, i.e.,
$\forall x,y,z\in M\colon d(x,z)\le d(x,y)+d(y,z)$.) Any finite metric space
can be described as the complete undirected graph whose vertices are $M$ with
an edge weight function $\{u,v\}\mapsto d(u,v)$. 
Conversely, we can interpret any connected, weighted, and undirected graph  
as the metric space whose point set is the vertex set and whose distance
function maps two points to the length of a shortest path between them. 

In particular, any connected, unweighted, and undirected graph also describes a
metric space in the same way.  From now on we always assume our graphs to be
finite, connected, undirected, and unweighted, unless specified otherwise.
Important metric spaces include the real line, the Euclidean plane, and
generally the Euclidean space of any given dimension.

For any positive integer $k$ and metric space $\mathcal{M}=(M,d)$, we call a map
$C\colon [k]\to M$ a $k$-\emph{configuration}.  In the context of the
$k$-server problem we may describe a $k$-configuration by saying that server
$s_i$ is at point $C(i)$ for any $i\in[k]$.  
We say that a $k$-configuration $C$ \emph{covers} or \emph{occupies} a point
$p\in M$ with server $s_i$ if $C(i)=p$ for some $i\in[k]$.  For a sequence of
$k$-configurations $C_0,\dots,C_n$, we say, for any $i\in[k]$ and $j\in[n]$,
that server $s_i$ moves from $C_{j-1}(i)$ to $C_j(i)$ in step $j$. 
Whenever the value of $k$ is clear from the context, which is the case for the
remainder of this paper, we omit the $k$ and simply say configuration instead
of $k$-configuration.  An \emph{instance} of the $k$-server problem on
$\mathcal M$ of length $n$ is a sequence $(r_1,\dots,r_n)$ of $n$ points in $M$
and an initial configuration $C_0\colon [k]\to M$. 
A \emph{solution} is a sequence of configurations $(C_1,\dots,C_n)$ such that
$C_j$ covers $r_j$ for all $j\in[n]$.  The \emph{cost} of such a solution is
$\sum_{i=j}^n\sum_{i=1}^k d(C_{j-1}(i),C_j(i))$ in the distance model and
$\sum_{j=1}^n\max\{d(C_{j-1}(i),C_j(i))\mid i\in[k]\}$ in the time model. 

When we speak of the $k$-server problem without further specification, this
refers to the time model for the remainder of this paper. 

\paragraph*{Intricacies of the Time Model.}
In the distance model, any algorithm can without loss of generality be
transformed into one that is \emph{lazy}, i.e., moving only one server at a
time. Hence it is common and convenient to argue only about lazy algorithms.
In the time model, this simplification is not justified anymore; we also need
to consider movements of the servers other than the one serving a request.
This leads to intricacies peculiar to the time model, which we briefly discuss
here.

Consider a graph with integer weights.  Suppose that a server incident to an
edge of unit length moves across this edge to serve a request at the other end.
Suppose that another server is positioned at a vertex incident to a longer edge of length $2$.
Depending on the situation to be modeled, one might want to allow or disallow a
synchronous movement by the second server to the midpoint of the longer edge.
In the default formulation, servers can only traverse edges completely; partial
movements are impossible. 
This implies the following perhaps unintuitive behavior: Even if a server that
traversed a unit edge to serve a request then moves back to serve the next
request, traveling a length of $2$ in two steps, it is impossible for another
server to traverse an edge of length $2$ at the same time without extra cost.
Even if the servers move synchronously as far as possible, the cost is at least
$3$.  This problem does not occur in unweighted graphs, however. Indeed, if it
is desired to enable the servers to traverse a long edge over multiple steps,
this is easily achieved by transforming a given graph with rational weights
into an unweighted one: it suffices to subdivide the edges into segments of a
length dividing all occurring weights. 
(The size of the resulting graph varies with the weights, of course.)

All of our results apply to both variants since our algorithm works also on
weighted graphs, and the graphs constructed for our lower bounds are all
unweighted.

\section{Results}

We now present our results for the time-optimal $k$-server problem: The upper
bounds in \cref{sec:upperbounds}, the lower bounds for deterministic algorithms
in \cref{sec:lowerboundsdeterministic}, and a lower bound for randomized
algorithms in  \cref{sec:lowerboundsrandomized}. 

\subsection{Upper Bounds}\label{sec:upperbounds}

We start with a positive result and show that the following algorithm \sequen for the
time-optimal $k$-server problem obtains a competitive ratio of $\Delta k$, 
where the aspect ratio $\Delta$ of a metric space is its diameter divided by
the minimal distance between distinct points.

\begin{definition}[{Algorithm \sequenbold}] 
  Label the $k$ servers $s_1,\dots,s_k$ in arbitrary order. If a requested
  point is already covered by a server, then \sequen does not change its
  configuration.  The $m$-th request that requires \sequen to move a server is 
  served by $s_i$ with $i\equiv_k m$, while all other servers stay idle.
  In other words, \sequen moves a server only when it is necessary, only one
  server at a time, the first $k$ movements are by $s_1,\dots,s_k$ in this
  order, and then it repeats cyclically.
\end{definition}

Note that \sequen never moves servers synchronously, thus the cost is the same
in the time model and the distance cost. Essentially, \sequen is a marking
algorithm~\cite{BE1998,Kom2016}.

\begin{theorem}
  \sequen is $\Delta k$-competitive on metric spaces with aspect ratio~$\Delta$. 
\end{theorem}

\begin{proof}
    Without loss of generality, let the minimal distance between two distinct
    points be $1$; the diameter of the given metric space is therefore
    $\Delta$.  We may assume that the servers of \sequen and an optimal
    algorithm \Opt start in the same configuration. We split the instance
    into $\ell$ disjoint phases. 
    \begin{itemize} 
      \item The first phase $P_1$ starts with the first request. 
      Without loss of generality, we assume it to be for a point outside of the initial configuration. 
      \item A phase ends as soon as $s_k$ has moved or if no request is left.
      \item If phase $P_i$ has ended and a request is left, the next phase $P_{i+1}$ begins immediately.
    \end{itemize}
    
    Observe that in every phase \sequen moves each of the servers
    $s_1,\dots,s_k$ exactly once and therefore incurs cost of at most
    $\Delta\cdot k$ in each phase. 
    
    We show that \Opt must have a cost of at least $1$ in each phase $P_i$ for $i\in[\ell-1]$: 
    For the first phase $P_1$ the statement is trivial: the first request
    requires both \alg and \Opt to move. We may thus assume $i\ge2$.  Note
    that each phase $P_i$ contains at least $k$ distinct requested points. 
    Otherwise, \sequen would serve the first up to $k-1$ distinct requested
    points with servers $s_1,\dots,s_{k-1}$ and would never move $s_k$.  The
    phase could thus end only when no request is left, which is possible only
    for the final phase $P_\ell$. 
    Let $i\in[\ell-1]$ and denote by $x_i$ the point requested in phase
    $P_{i-1}$ for which \sequen moved server $s_k$ to serve it, which ended
    phase $P_{i-1}$.  The optimal algorithm \Opt must also serve this
    request and thus have a server at
    $x_i$ at the end of phase $P_{i-1}$ and thus the start of phase $P_i$. 
    Now consider the $k$ requests of phase $P_i$ for
    which \sequen moves a server. Since $s_k$ only moves for the last one of
    them, none of these $k$ requests can be $x_i$; thus there are $k$ distinct
    requested locations in phase $P_i$ that are all distinct from $x_i$.
    Therefore, in $P_i$ \Opt must move the server from $x_i$ or use
    the other $k-1$ servers to serve requests at $k$ distinct locations, which
    implies that at least one server has to serve requests at two distinct
    locations and thus move, causing a cost of $1$ as claimed.
    
    We now know that each configuration change causes a cost of at least $1$ for \Opt. Thus, \Opt
    incurs a cost of at least $\ell-1$. This yields the desired bound 
    $\cost{\sequen,I}\leq \ell\cdot \Delta\cdot k\leq \Delta\cdot k\cdot
    \cost{\Opt(I)}+\Delta\cdot k$, which concludes the proof.
\end{proof}

\begin{corollary}\label{corollary:uniformD}\label{corollary:uniform}
  On graphs of diameter $D$ or less, \sequen is $Dk$-competitive.
  In particular, \sequen is $k$-competitive on any uniform metric space.
\end{corollary}

\subsection{Lower Bounds for Deterministic Algorithms}\label{sec:lowerboundsdeterministic}

In this section, we provide lower bounds for deterministic algorithms on various metric spaces. 

\paragraph*{\boldmath Universal Lower Bound of $k$ on Almost All Graphs.}
We begin with a bound of $k$ for all unweighted graphs on which the problem is
nontrivial. 

\begin{theorem}\label{thm:lowerboundk}
  Let $k$ be a positive integer. No algorithm for the time-optimal $k$-server
  problem can be better than $k$-competitive on graphs with more than $k$
  vertices. 
\end{theorem}

\begin{proof}
  Let $k$, $G$, and \alg as described be given. 
  Choose any connected set $S$ of $k+1$ vertices of $G$ and consider the metric
  subspace described by the induced subgraph $G[S]$.  We call a configuration
  where all $k$ servers occupy pairwise distinct locations within $S$ a
  \emph{distinct $S$-configuration}.
  Observe that any given distinct $S$-configuration $C$ can be transformed into
  any other given distinct $S$-configuration $C'$ at cost at most $1$. This can
  be done by choosing a path inside $G[S]$ from its vertex not covered by $C$
  to its vertex not covered by $C'$ and moving all servers by one edge along
  this path. 

  Therefore, it is always possible to partition any instance of a length
  divisible by $k$ into consecutive subsequences of length exactly $k$ such
  that an optimal algorithm incurs a cost of only $1$ per subsequence.  In
  particular, this holds for such an instance constructed by always
  requesting a point of $S$ currently not covered by \alg. On such an
  instance, \alg incurs a cost of $1$ with every request, whereas the optimum
  can serve $k$ requests at cost at most $1$ by moving as described to a
  configuration that covers the currently requested vertex and those requested
  in the following $k-1$ steps. 
\end{proof}

\paragraph*{\boldmath Lower Bound of $k+1$ on the Line.}
We consider now a special type of space that has received a lot of attention
(see \cref{subsec:related}) for the classical $k$-server problem: the line.
As mentioned above, we have $k$-competitive algorithms for $k$-server in the
distance model on the real line, the rational line, and the integer line.  The
following theorem shows that the time-optimal $k$-server problem is
strictly harder in all of these cases.

\begin{theorem}\label{thm:kplusoneline}
  For any integer $k\geq 2$, no algorithm for the time-optimal $k$-server
  problem has a competitive ratio better than $k+1$ on the following spaces:
  the real line, the rational line, the integer line, any finite unweighted
  cycle with an even number of at
  least $2k+6$ vertices, and the continuous one-dimensional sphere $\sphere^1$. 
\end{theorem}

\begin{proof}
  Let \alg be any online algorithm for the given metric space.  For a finite
  unweighted cycle with an even number of $N\geq 2k+6$ vertices, label the
  vertices counterclockwise by $0,1,\dots, N-1$. For the sphere $\sphere^1$,
  choose $N=2k+6$ points evenly spaced along this metric space, label them
  $0,1,\dots,N-1$ counterclockwise, and re-scale the space such that the
  distance between consecutive points is $1$. 
  For simplicity, we refer to the labeled points as integers and say that an
  integer \emph{neighbors} another one if they are at distance $1$ from each other.
  In all cases, every even integer neighbors two odd ones and each odd one two
  even ones.
  We will construct an instance phase by phase such that there is an optimal
  algorithm \Opt satisfying the following invariant:

  \begin{claim}\label{clm:invariantconsecutive}
    At the beginning of each phase, the $k$ servers of \Opt occupy $k$
    distinct even integers or $k$ distinct odd integers. Moreover, at least two
    of these integers are consecutive (in the sense of being separated by a
    distance of exactly $2$).
  \end{claim}
  
  Without loss of generality we assume for the following description that the
  integers occupied by \Opt at the start of the phase are all even by
  choosing an appropriate initial configuration and shifting the labels of the
  integers by one after each phase.  Any phase will consist of potentially
  repeated requests for $k$ distinct odd integers, among which at least two are
  consecutive. Moreover, the requests will be chosen such that \Opt can
  cover all of them by moving all of its $k$ servers simultaneously, each by a
  distance $1$ and thus at total cost $1$, at the beginning of the phase and no
  movements afterwards. Any phase constructed with these properties retains the
  invariant for the next phase, 
  allowing us to iterate the process to construct an arbitrarily long instance
  of arbitrarily high optimal cost. 

  Since \Opt does not move its servers anymore after their phase-initial
  synchronous movements, we can request any point covered by \Opt as many
  times as we would like without increasing the cost of \Opt or changing
  \Opt's configuration at the end of the phase. We can thus enforce the
  following invariant for \alg without loss of generality: 
  
  \begin{claim}\label{clm:stayentirephase}
    Once \alg has served a request at some point during some phase, it will
    always have a server at this point during all further requests of this
    phase.
  \end{claim}

  \begin{claimproof}[Proof of claim]
  Assume that \alg makes a move such that some integer previously requested in
  the current phase is no longer covered. Then the constructed instance could
  be extended by immediately requesting this integer again after this move, at
  no additional cost to \Opt. This can be continued until all the integers
  previously requested in the current phase are covered again by \alg.
  \end{claimproof}
  
  Using this invariant of \alg, we can now also assume the following without
  loss of generality: 

  \begin{claim}\label{clm:eachoddonce}
  In each phase, $k$ distinct odd integers are requested, each exactly once.
  \end{claim}
  
  In summary, we can assume the following properties for each phase: 

  \begin{enumerate}
    \item At the beginning of the phase, the servers of \alg occupy the same positions as the servers of \Opt, namely $k$ \emph{distinct} even integers, at least two of which are consecutive (in the sense of being separated by a distance of $2$). 
    \item During the phase, $k$ distinct odd integers are requested one by one. 
      Each such integer must, once requested, be covered by \alg during all subsequent requests of the phase. 
    \item Moreover, there is a bipartite matching of weight exactly $1$ between the $k$ odd requested points and the $k$ phase-initially covered even integers. 
    \item Finally, \Opt moves its servers along this matching at cost $1$ when the first request of the phase appears and does not move its server anymore for the rest of the phase. 
  \end{enumerate}
  
  We now describe how the $k$ requests for any given phase are chosen depending
  on \alg's behavior in a way that guarantees the invariant of
  \cref{clm:invariantconsecutive} to hold for the following phase.

  Consider any phase. Partition the phase-initial server positions (which are
  the same for \Opt and \alg) into $\ell\geq 1$ maximal, pairwise disjoint,
  nonempty sets $S_1,\dots,S_\ell$ of consecutive even integers.   We define
  $k_i\coloneqq\abs{S_i}$ and note that $\sum_{i=1}^\ell k_i=k$. 
  Given a metric space $\mathcal M=(M,d)$, a point $c\in M$, and radius
  $\rho\ge 0$, we call $D_\rho(c)\coloneqq\{p\in M\mid d(p,c)<\rho\}$ the
  \emph{open ball} and $D_\rho[c]\coloneqq\{p\in M\mid d(p,c)\le\rho\}$ the \emph{closed ball}.
  For any $m\in[\ell]$, we call $R_m\coloneqq\bigcup_{i\in S_m}D_2(i)$ the \emph{range}
  of $S_m$; \Cref{fig:subdivision} shows an example for $k=5$. 

  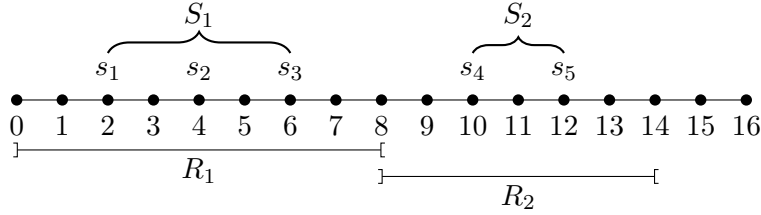
\begin{figure}
      \begin{center}
        \begin{tikzpicture}[scale=0.6, graphnode/.style={fill, circle, inner sep=1.5pt}]
          \draw (-1,0)--(0,0)--(1,0)--(2,0)--(3,0)--(4,0)--(5,0)--(6,0)--(7,0)--(8,0)--(9,0)--(10,0)--(11,0)--(12,0)--(13,0)--(14,0)--(15,0);
          \node[graphnode, label=below:{0}] at (-1,0) {};
          \node[graphnode, label=below:{1}] at (0,0) {};
          \node[graphnode, label=below:{2}] at (1,0) {};
          \node[graphnode, label=below:{3}] at (2,0) {};
          \node[graphnode, label=below:{4}] at (3,0) {};
          \node[graphnode, label=below:{5}] at (4,0) {};
          \node[graphnode, label=below:{6}] at (5,0) {};
          \node[graphnode, label=below:{7}] at (6,0) {};
          \node[graphnode, label=below:{8}] at (7,0) {};
          \node[graphnode, label=below:{9}] at (8,0) {};
          \node[graphnode, label=below:{10}] at (9,0) {};
          \node[graphnode, label=below:{11}] at (10,0) {};
          \node[graphnode, label=below:{12}] at (11,0) {};
          \node[graphnode, label=below:{13}] at (12,0) {};
          \node[graphnode, label=below:{14}] at (13,0) {};
          \node[graphnode, label=below:{15}] at (14,0) {};
          \node[graphnode, label=below:{16}] at (15,0) {};
          \node[label=above:{$s_1$}] at (1,0) {};
          \node[label=above:{$s_2$}] at (3,0) {};
          \node[label=above:{$s_3$}] at (5,0) {};
          \node[label=above:{$s_4$}] at (9,0) {};
          \node[label=above:{$s_5$}] at (11,0) {};
          \node[above=10pt] (s1) at (1,0) {};
          \node[above=10pt] (s3) at (5,0) {};
          \node[above=10pt] (s4) at (9,0) {};
          \node[above=10pt] (s5) at (11,0) {};
          \node[above=16pt, label=above:{$S_1$}] at (3,0) {};
          \node[above=16pt, label=above:{$S_2$}] at (10,0) {};
          \node[below=15pt] (r1) at (-1.3,0) {};
          \node[below=15pt] (r2) at (7.3,0) {};
          \node[below=35pt, label={$R_1$}] at (3,0) {};
          \node[below=25pt] (r3) at (6.7,0) {};
          \node[below=25pt] (r4) at (13.3,0) {};
          \node[below=45pt, label={$R_2$}] at (10,0) {};
          \draw[{]-[}] (r1)--(r2);
          \draw[{]-[}] (r3)--(r4);
          \draw[decoration={brace,amplitude=8}, decorate, thick] (s1.north) -- (s3.north);
          \draw[decoration={brace,amplitude=6}, decorate, thick] (s4.north) -- (s5.north);
        \end{tikzpicture}
      \end{center}
      \caption{Example of the subdivision of phase-initial server positions and
      the corresponding ranges used in the proof of \cref{thm:kplusoneline}.
      (Note that the figure does not show the entire space, which is a cycle or
      the infinite line, but only a segment.)}
      \label{fig:subdivision}
  \end{figure}

  Note that the range $R_m$ contains exactly $k_m+1$ odd integers. Exactly
  $k_m$ of them are requested during a phase in the instance family described
  below.  This guarantees that \Opt can indeed move all its servers from
  their phase-initial positions to the points request during this phase at cost
  $1$ and then keep them there for the remainder of the phase, which proves the
  first part of the invariant of \cref{clm:invariantconsecutive}. 
  We call the range $R_m$ \emph{saturated} if the instance has already
  requested $k_m$ out of the $k_m$ odd integers of $R_m$ during the current
  phase and \emph{unsaturated} otherwise.  Recall that once an integer is
  requested, a server occupies it during all remaining requests of the phase,
  thus any saturated $R_m$ contains at least $k_m$ servers. 
  A phase ends when all ranges are saturated. 
  
  Since there are two consecutive even integers in the phase-initial configuration
  of \Opt, by renaming
  we can assume without loss of generality that $k_1\geq 2$ and that $S_1=\{2,4,\dots, 2k_1\}$. 
  Since $R_1$ contains exactly $k_1+1\ge 3$ consecutive odd integers and $k_1$
  of them are requested during the phase, at least two consecutive odd
  integers are requested. This proves the second part of the invariant of
  \cref{clm:invariantconsecutive}.

  We now describe what the requests in the considered phase depend on. 
  For this, we will call a point \emph{distant} if and only if it has distance
  at least $1$ from all current servers, i.e., if it lies outside of all open
  unit balls around the current server locations.  The first request is $3$,
  which is indeed a distant point since all servers are at even integers
  phase-initially. 
  From then on, a further request is chosen according to the following
  \emph{selection procedure} until $k-1$ odd integers have been requested
  during the phase or the procedure becomes impossible: Request an arbitrary
  distant odd integer not previously requested during the phase 
  from any unsaturated range $R_m$ except for the integer $1$. 
  The sequence of such requests ends only when $k_m$ requests have been made in
  each range $R_m$ with $1\neq m\in [\ell]$ and $k_1-1$ requests have been made
  in $R_1$, or if no such request is possible anymore.  We will show that once
  this process ends, there must be a point in $R_1$ that no server can reach
  with cost less than $2$ using the following simple claim:

	\begin{claim}\label{claim:uncoveredpoints}
	  If a range $R_m$ for any $m\in[\ell]$ contains at most $k_m$ servers,
	  it also contains a distant odd integer. 
    If the \emph{truncated range} $R_1'\coloneqq R_1\setminus(0,4]$ contains at
    most $k_1-2$ servers, it contains a distant odd integer.
  \end{claim}
	
	\begin{claimproof}[Proof of claim]
    A unit ball can contain at most one odd integer. Moreover, unit balls
    around points outside of a range $R_m$ cannot contain odd integers inside
    this range, and analogously for the truncated range $R_1'$.  Hence, the
    range $R_m$ for any $m\in[\ell]$ containing at most $k_m$ servers means
    that at most $k_m$ out of the $k_m+1$ odd integers in it are not distant
    and thus that there is a distant odd integer in $R_m$. 
    Analogously, the truncated range $R_1'$ containing at most $k_1-2$ servers
    means that at most $k_1-2$ out of the $k_1-1$ odd integers in it are not
    distant, and thus that there is a distant odd integer in $R_1'$.
  \end{claimproof}
	
  By this claim, if the selection procedure ends, any range $R_m$ with $1\neq m
  \in[\ell]$ must contain at least $k_m$ servers, whether or not it is
  saturated, and $R_1'$ must contain at least $k_1-2$ servers.  We consider two
  different cases:

	\begin{itemize}[left=0em]
  \item \textit{Case 1: $R_1'$ contains at least $k_1-1$ servers.}
    In this
    case, $R_1$ contains at least $k_1$ servers (including the server positioned on $3$).
    This means that all ranges $R_m$ with $1\neq m\in [\ell]$ must contain
    exactly $k_m$ servers and must therefore be saturated. Thus, these servers
    all occupy odd integers and have a distance of at least $2$ to the odd
    integer $1$. The same holds for any servers in $R_1'$.
  \item \textit{Case 2: $R_1'$ contains exactly $k_1-2$ servers.}
    In this case, $R_1'$ contains a distant point. Since the selection
    procedure ended, a total of $k_1-1$ requests must have already been made to
    $R_1$ and served. Thus the $k_1-2$ servers in $R_1'$ occupy distinct odd
    integers and cannot reach any other odd integer in $R_1$ with cost less
    than $2$. Of the other ranges $R_m$ with $1\neq m\in[\ell]$, all but
    possibly one must be saturated, since they contain exactly $k_m$ servers.
    Therefore these servers occupy odd integers and cannot reach odd integers
    in $R_1$ with cost less than $2$. Finally, one range $R_m$ with $1\neq m\in
    [\ell]$ may contain $k_m+1$ servers, or a single server may occupy a point
    outside a range. Such a server can only reach either the odd integer $1$ or
    the odd integer $2k_1-1$ in $R_1$ with cost less than $2$. This is clear if
    the metric space is a line.

    If it is a cycle or the continuous sphere, the
    two ranges would have to meet at both ends. However, this would mean that
    they contain all odd integers in the metric space, of which there are at
    least $k+3$, while they contain a maximum of $(k_1+1)+(k_{m}+1)\leq k+2$
    odd integers. Thus, since there are two odd integers in $R_1$ that have not
    been requested yet, one of them cannot be reached by any server with cost
    less than $2$.
  \end{itemize}
	
  We thus have, in both cases, an odd integer in $R_1$ that is at distance at
  least $2$ from all servers.  Such an integer is requested next, forcing
  \alg to incur a cost of at least $2$ when serving it. 
  From then on, additional distant points in unsaturated ranges $R_m$ with
  $m\in [\ell]$ (including $R_1$) are requested. Such points exist by
  \Cref{claim:uncoveredpoints}. Once $k$ distinct odd integers have been
  requested, all ranges are saturated and the phase ends. 
	
  Recall that the first request was to the odd integer $3$. Since $R_1$
  contains $k_1+1$ odd integers, $k_1$ of which have been requested, at least
  one of the odd integers $1$ or $5$ must have been requested, and thus indeed
  at least two consecutive odd integers.  This means that the final
  configuration of the phase is such that the invariant of
  \cref{clm:invariantconsecutive} is satisfied for an immediately following
  phase.

  We can thus iterate the process and construct arbitrarily many new phases
  such that for each phase \alg incurs a cost of at least $k+1$ while
  \Opt has a cost of exactly $1$.
\end{proof}

Note that \cref{thm:kplusoneline} generalizes the lower bound of $3$ given by
Koutsoupias and Taylor for the special case of $k=2$ servers on the real
line~\cite[Thm.~4]{KT2004}.  Furthermore, our result also implies
hardness for many other spaces, e.g., the Euclidean plane, as shown in the
following corollary.

\begin{corollary}\label{corollarykplusoneline}
  No algorithm can solve the time-optimal $k$-server problem with a competitive
  ratio better than $k+1$ in the Euclidean space of any positive dimension.
\end{corollary}

\begin{proof}
  Identify the real line with the first axis of any given Euclidean space of
  positive dimension, and consider the induced hard instance family showing
  that no algorithm can be better than $k+1$ competitive for the real line
  from the proof of \cref{thm:kplusoneline}.

  Assume towards contradiction that there is an algorithm better than
  $(k+1)$-competitive on the considered Euclidean space. 
  Then the projection of this algorithm on the first dimension (obtained by
  replacing each server position computed by the algorithm for the Euclidean
  space by its projection onto the first axis) cannot incur greater cost
  because the distance between the projections of two points is never larger
  than the distance between the two points themselves. Thus this projected
  algorithm would be better than $(k+1)$-competitive on the induced hard
  instance family, which yields the desired contradiction.
\end{proof}

\paragraph*{\boldmath Lower Bound of $3k/2$ on Small Graphs.}
While the lower bound of $k+1$ for the time model from \cref{thm:kplusoneline}
improves over the best known lower bound of $k$ for the distance model, it does
so just barely.  Although this is sufficient to show that the direct analogue
to the $k$-server conjecture is not true in the time model, we are able to
improve the lower bound further to $3k/2$ on instances that are only slightly
larger. 

We first provide in \cref{lem:doublecircle} fixed-size instance segments, for
which an algorithm with two servers spends at least $3$ times the cost of an optimal
solution. The instances are designed to be iterable and thus immediately yield
a lower bound of $3k/2$ on the competitive ratio for $k=2$ in
\cref{cor:doublecircle}.  The underlying space for these constructions is
called the \emph{double cycle}; see \cref{fig:doublecircle}.  
This space is also used as a gadget to in the proof of \cref{thm:threehalfk},
which extends the statement to any even $k\ge2$. 

\begin{definition}\label{def:doublecircle}
  Consider the graph on the twelve vertices $A_i$ and $B_i$ for $i\in[6]$, and the
  edge sets
  $\{\{A_i,A_{i+1}\},\{A_i,B_{i+1}\},\{B_i,A_{i+1}\},\{B_i,B_{i+1}\}\mid i\in[6]\}$,
  where we define $A_7\coloneqq A_1$ and $B_7\coloneqq B_1$.  We call this graph, which
  is depicted in \cref{fig:doublecircle}, the \emph{double cycle}. 
\end{definition}

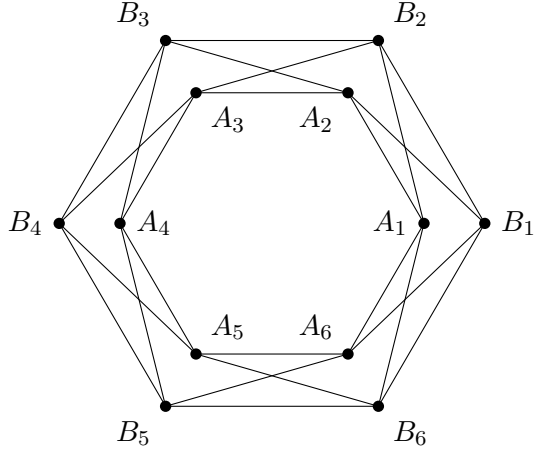
\begin{figure}
  \begin{center}
    \begin{tikzpicture}
      \def\hexstretch{1.4};
      \coordinate (a1) at (2,0);
      \coordinate (a2) at (1,1.7320508075688772);
      \coordinate (a3) at (-1,1.7320508075688774);
      \coordinate (a4) at (-2,0);
      \coordinate (a5) at (-1,-1.7320508075688767);
      \coordinate (a6) at (1,-1.732050807568878);
      \coordinate (b1) at (\hexstretch*2,\hexstretch*0);
      \coordinate (b2) at (\hexstretch*1,\hexstretch*1.7320508075688772);
      \coordinate (b3) at (-\hexstretch*1,\hexstretch*1.7320508075688774);
      \coordinate (b4) at (-\hexstretch*2,\hexstretch*0);
      \coordinate (b5) at (-\hexstretch*1,-\hexstretch*1.7320508075688767);
      \coordinate (b6) at (\hexstretch*1,-\hexstretch*1.732050807568878);
      \node[fill, circle, inner sep=1.5pt, label=left:{$A_1$}] at (a1) {};
      \node[fill, circle, inner sep=1.5pt, label=below left:{$A_2$}] at (a2) {}; 
      \node[fill, circle, inner sep=1.5pt, label=below right:{$A_3$}] at (a3) {}; 
      \node[fill, circle, inner sep=1.5pt, label=right:{$A_4$}] at (a4) {}; 
      \node[fill, circle, inner sep=1.5pt, label=above right:{$A_5$}] at (a5) {}; 
      \node[fill, circle, inner sep=1.5pt, label=above left:{$A_6$}] at (a6) {};
      \node[fill, circle, inner sep=1.5pt, label=right:{$B_1$}] at (b1) {};
      \node[fill, circle, inner sep=1.5pt, label=above right:{$B_2$}] at (b2) {}; 
      \node[fill, circle, inner sep=1.5pt, label=above left:{$B_3$}] at (b3) {}; 
      \node[fill, circle, inner sep=1.5pt, label=left:{$B_4$}] at (b4) {}; 
      \node[fill, circle, inner sep=1.5pt, label=below left:{$B_5$}] at (b5) {}; 
      \node[fill, circle, inner sep=1.5pt, label=below right:{$B_6$}] at (b6) {};    
      \draw (a1)--(a2)--(a3)--(a4)--(a5)--(a6)--(a1);
      \draw (b1)--(b2)--(b3)--(b4)--(b5)--(b6)--(b1);
      \draw (a1)--(b2)--(a3)--(b4)--(a5)--(b6)--(a1);
      \draw (b1)--(a2)--(b3)--(a4)--(b5)--(a6)--(b1);
    \end{tikzpicture}
  \end{center}
  \caption{The double cycle metric space used in the proof of \cref{lem:doublecircle}.}
  \label{fig:doublecircle}
\end{figure}

The following lemma describes a process that can be repeated arbitrarily often
to yield a lower bound of $3$ on the competitive ratio in the case of $k=2$
servers.

\begin{lemma}\label{lem:doublecircle}
  Consider the $2$-server problem on the double cycle defined in
  \cref{def:doublecircle} and depicted in \cref{fig:doublecircle}, and let
  \alg be any online algorithm.  Assume that the two servers of \Opt have
  some positions $p_1$ and $p_2$ satisfying $p_1\in\{A_i,B_i\}$ and
  $p_2\in\{A_{j},B_{j}\}$ for some $j\equiv_6i+2$, whereas \alg might have
  arbitrary initial server positions. Then there is 
  a sequence of two requests such that
  \begin{enumerate}[label=(\roman*)]
    \item \alg incurs a cost of at least $3$ serving them, 
    \item \Opt incurs a cost of exactly $1$ serving them, and 
    \item the positions $q_1$ and $q_2$ of the servers of \Opt at the end of
    the instance satisfy $q_1\in\{A_i,B_i\}$ and $q_2\in\{A_{j},B_{j}\}$ for
    some $j\equiv_6i+2$.
  \end{enumerate}
\end{lemma}

\begin{corollary}\label{cor:doublecircle}
  There is no algorithm for the 2-server problem in the time model on the double
  cycle with a competitive ratio better than $3$.
\end{corollary}

\begin{proof}[Proof of \cref{lem:doublecircle}]
  By symmetry, we may assume that the initial positions of the two servers of
  \Opt are $B_2$ and $B_4$. We will consider two distinct cases and prove
  the first property of the claim in both of them. 
  
  \begin{itemize}[left=0em]
    \item \textit{Case 1.} 
    Assume that there is at least one point in $\{A_1,B_1,A_3,B_3,A_5,B_5\}$
    that is not contained in the union of the closed neighborhoods of the
    positions initially occupied by \alg's servers. Such a point is then
    requested first and \alg incurs a cost of at least $2$ serving it. If the
    request was any point $\{A_1,B_1,A_5,B_5\}$, then the next request is
    either $A_3$ or $B_3$, whichever is not occupied by \alg. (They cannot
    both be occupied by \alg because one of its servers is still at one of the
    points in $\{A_1,B_1,A_5,B_5\}$.) If the request was to $A_3$ or $B_3$, an
    unoccupied point among $\{A_1,B_1,A_5,B_5\}$ is requested. This forces \alg
    to incur a further cost of~$1$. 
    \item \textit{Case 2.} 
    Assume that the union of the closed neighborhoods of the positions
    initially occupied by \alg's servers contain all of
    $\{A_1,B_1,A_3,B_3,A_5,B_5\}$. The closed neighborhood of a point in this
    set contains only one such point and the closed neighborhood of a point in
    the complement contains only four such points. To cover all six points, the
    two positions must thus be two points in the complement.  They can also not
    be $A_i$ and $B_i$ for the same $i$ since this would again cover only four
    of the six points. 
    We thus know that the two server positions of \alg are obtained by choosing
    two of the pairs $\{A_2,B_2\}$, $\{A_4,B_4\}$, and $\{A_6,B_6\}$ and then
    one point of each chosen pair. Without loss of generality, let the chosen
    points be $B_2$ and $B_4$.  Then $A_3$ or $B_3$ is requested first.
    \begin{itemize}
      \item \textit{Case 2a.} If \alg serves this request and moves the other
      server to $A_6$ or $B_6$, it has cost $2$. Then, any point among
      $\{A_1,B_1,A_5,B_5\}$ is requested, which forces \alg to incur cost at least
      $1$.
      \item \textit{Case 2b.} If \alg serves this request and does not
      move the other server to $A_6$ or $B_6$, it has cost at least $1$.  Then,
      \alg can serve at most two points among $\{A_1,B_1,A_5,B_5\}$ with cost
      less than $2$, and thus there are at least two points where this is not the
      case. Such a point is requested, letting \alg additionally incur a cost of
      at least $2$.
    \end{itemize}
  \end{itemize}
  
  We have seen that \alg indeed incurs a cost of at least $3$ for the two
  requests in both cases.  We observe that in all cases, the two points to be
  requested can be chosen as follows: Out of the three pairs $\{A_1,B_1\}$,
  $\{A_3,B_3\}$, and $\{A_5,B_5\}$, two pairs are chosen, and then from each
  chosen pair, one point is chosen. This already guarantees the last claimed
  property. 
  
  Finally, we note that \Opt can for any choice serve both requests by
  moving its two servers once by distance $1$ simultaneously: Both servers
  moving clockwise for one of the three possible choices of pairs, both
  counterclockwise for another choice, and one clockwise and the other
  counterclockwise for the last choice.  This proves the second claim. 
\end{proof}

We now generalize the construction of the double cycle to
\emph{double cycle chains}, which allows us to prove a lower bound of $3k/2$ on
the competitive ratio for an arbitrarily large even number $k$ of servers.

\begin{definition}\label{def:kcircles}
  Let an arbitrary positive even integer $k$ be given. 
  Consider the graph that is constructed by taking $k/2$ copies $G_i$ with
  $i\in[k/2]$ of the double cycle depicted in \cref{fig:doublecircle} and linking
  them by adding, for every $i\in[k/2-1]$, a path on $4$ new vertices and
  connecting by an edge one of its ends with the vertex corresponding to $B_1$ of
  $G_i$ and by another edge at the other end with the vertex corresponding to
  $B_4$ of $G_{i+1}$.  We call this graph, which is shown in
  \cref{fig:kcircles}, the \emph{double cycle chain} of length $k/2$.
  It contains exactly 
  $2\cdot6\cdot k/2 + 4(k/2-1) = 8k - 4\in\mathcal{O}(k)$ 
  vertices and 
  $4\cdot6\cdot k/2 + 5(k/2-1) = 29k/2 - 5\in\mathcal{O}(k)$ 
  edges. 
\end{definition}

\begin{theorem}\label{thm:threehalfk}
  For any even integer $k\geq 2$, there is no online algorithm for
  the time-optimal $k$-server problem on the double cycle chain of length $k/2$
  defined in \cref{def:kcircles} and depicted in \cref{fig:kcircles} with a
  competitive ratio better than $3k/2$.
\end{theorem}

\begin{proof}
  Fix any even integer $k\ge2$. Let \alg be any online algorithm for the
  time-optimal $k$-server problem on the double cycle chain of length $k/2$.
  Assume without loss of generality that \Opt starts with two servers in
  each double-cycle gadget, each in a configuration equivalent to the starting
  position described in \cref{lem:doublecircle}. We construct an instance phase
  by phase.
  
  In each phase, exactly two distinct points that are a distance $2$ apart are
  requested from each of the $k/2$ double cycles. 
  Specifically, each phase is constructed by applying \cref{lem:doublecircle}
  once to each double cycle.  We make the following assumption, which is the
  same as \Cref{clm:stayentirephase} in the proof of \Cref{thm:kplusoneline}
  and can be shown using the same argument.

  \begin{claim}\label{clm:linearstayentirephase}
    Once \alg has served a request at some position during some phase, it will
    always have a server at this position during all further requests of this
    phase. 
  \end{claim} 

  Then \Opt can still
  move its servers to the $k$ requested points with a single move of cost $1$
  at the beginning of the phase and stay there for the remainder of the phase.
  But \alg now incurs a cost of at least $3k/2$ in every phase: If the servers
  of \alg stay in the double cycle where they start during a phase, we
  immediately obtain this since we generate a cost of at least $3$ for \alg in
  every double cycle by \cref{lem:doublecircle}. 
  
  We now show that \alg cannot gain anything by moving servers between double cycles. 
  We accomplish this by ordering the pairs of requests such that for the next
  pair we always choose a double cycle where no point has been requested yet
  during the phase and where \alg has at most two servers positioned on the
  double cycle extended by the two vertices of the closer half of each path
  connected to it.  Such a choice is always possible by the pigeonhole
  principle and because, 
  due to \cref{clm:linearstayentirephase}, two servers are positioned in each
  double cycle where points have been requested already during the phase.
  Then, \alg cannot improve on the cost of $3$ by using any other server for
  these two requests, as any such server is at least a distance of $3$ away
  from any point of $G_i$.
  
  By \cref{lem:doublecircle}, at the end of each phase the servers of \Opt
  are in a position that allows us to construct a new phase. We thus obtain the
  lower bound of $3k/2$ on the competitive ratio.
\end{proof}

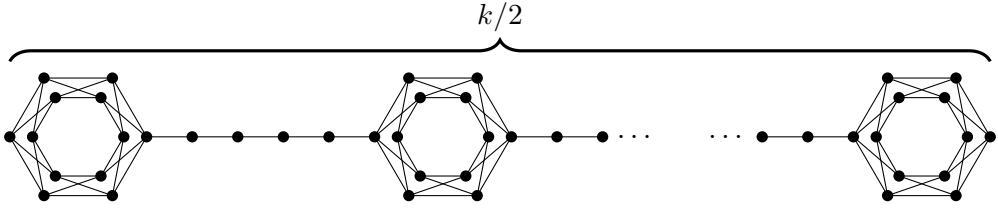
\begin{figure}
    \begin{center}
      \begin{tikzpicture}[scale=0.3, graphnode/.style={fill, circle, inner sep=1.5pt}]
      \def\hexstretch{1.5};
      \newcommand{\doublecircle}{\
      \coordinate  (a1) at (2,0);
      \coordinate  (a2) at (1,1.7320508075688772);
      \coordinate  (a3) at (-1,1.7320508075688774);
      \coordinate  (a4) at (-2,0);
      \coordinate  (a5) at (-1,-1.7320508075688767);
      \coordinate  (a6) at (1,-1.732050807568878);
      \coordinate  (b1) at (\hexstretch*2,\hexstretch*0);
      \coordinate  (b2) at (\hexstretch*1,\hexstretch*1.7320508075688772);
      \coordinate  (b3) at (-\hexstretch*1,\hexstretch*1.7320508075688774);
      \coordinate  (b4) at (-\hexstretch*2,\hexstretch*0);
      \coordinate  (b5) at (-\hexstretch*1,-\hexstretch*1.7320508075688767);
      \coordinate  (b6) at (\hexstretch*1,-\hexstretch*1.732050807568878);
      \node[graphnode] at (a1) {};
      \node[graphnode] at (a2) {}; 
      \node[graphnode] at (a3) {}; 
      \node[graphnode] at (a4) {}; 
      \node[graphnode] at (a5) {}; 
      \node[graphnode] at (a6) {};
      \node[graphnode] at (b1) {};
      \node[graphnode] at (b2) {}; 
      \node[graphnode] at (b3) {}; 
      \node[graphnode] at (b4) {}; 
      \node[graphnode] at (b5) {}; 
      \node[graphnode] at (b6) {};    
      \draw (a1)--(a2)--(a3)--(a4)--(a5)--(a6)--(a1);
      \draw (b1)--(b2)--(b3)--(b4)--(b5)--(b6)--(b1);
      \draw (a1)--(b2)--(a3)--(b4)--(a5)--(b6)--(a1);
      \draw (b1)--(a2)--(b3)--(a4)--(b5)--(a6)--(b1);}
      \begin{scope}
        \doublecircle
        \coordinate (c1) at (\hexstretch*2+2,0);
        \coordinate (c2) at (\hexstretch*2+4,0);
        \coordinate (c3) at (\hexstretch*2+6,0);
        \coordinate[above=4pt] (d0) at (\hexstretch*2,0);
        \coordinate[above=4pt] (dk) at (\hexstretch*2+8,0);
        \node[graphnode] at (c1) {};
        \draw (b1)--(c1);
        \node[graphnode] at (c2) {};
        \draw (c1)--(c2);
        \draw (c2)--(c3);
      \end{scope}
      \begin{scope}[xshift=\hexstretch*4cm+10cm]
        \doublecircle
        \node[graphnode] at (-\hexstretch*2-4,0) {};
        \node[graphnode] at (-\hexstretch*2-2,0) {};
        \node[graphnode] at (\hexstretch*2+2,0) {};
        \node[graphnode] at (\hexstretch*2+4,0) {};
        \draw (-\hexstretch*2-2,0)--(-\hexstretch*2-4,0);
        \draw (b4)--(-\hexstretch*2-2,0);
        \draw (b1)--(\hexstretch*2+2,0);
        \draw (\hexstretch*2+2,0)--(\hexstretch*2+4,0);
      \end{scope}
      \node[label=center:{\dots}] at (\hexstretch*5+17,0) {};
      \node[label=center:{\dots}] at (\hexstretch*5+21,0) {};
      \begin{scope}[xshift=\hexstretch*6cm+28cm]
      \doublecircle
      \node[graphnode] at (-\hexstretch*2-4,0) {};
      \node[graphnode] at (-\hexstretch*2-2,0) {};
      \draw (-\hexstretch*2-2,0)--(-\hexstretch*2-4,0);
      \draw (b4)--(-\hexstretch*2-2,0);
      \end{scope}
      \coordinate[above=1cm] (e0) at (-\hexstretch*2,0);
      \coordinate[above=1cm] (ek) at (\hexstretch*8+28,0);
      \draw[decoration={brace,amplitude=8}, decorate, very thick] (e0.north) -- (ek.north);
      \node[above=1cm,label={$k/2$}] at (\hexstretch*3+14,0) {};
      \end{tikzpicture}
    \end{center}
    \caption{The \emph{double cycle chain} of length $k/2$ used in \cref{thm:threehalfk}.}
    \label{fig:kcircles}
\end{figure}

\paragraph*{\boldmath Lower Bound of $5k/4$ on the Line for the Strict Competitive Ratio} 
We now provide the mentioned bound on the strict competitive ratio for
deterministic algorithms on the line.

\begin{theorem}\label{thm:strictlowerline}
	Let any positive even integer $k$ be given. 
  No deterministic algorithm for the time-optimal $k$-sever problem on the real
  line has a strict competitive ratio better than $5k/4$.
\end{theorem}

\begin{proof}
  We group the servers $s_1,\dots,s_k$ into $k/2$ pairs; for each $i\in[k/2]$, we call
  $\{s_{2i-1},s_{2i}\}$ pair $i$.  For each pair, we position its two servers a
  distance of $2$ apart and 
  separate the pairs by large distances $d\ge 2k+4$ between all neighboring pairs. 
  Specifically, we can choose the initial positions of servers $s_{2i-1}$ and
  $s_{2i}$ to be $id-1$ and $id+1$ for each $i\in[k/2]$. 
  
  We now construct a hard instance consisting of $k$ requests, which
  contains exactly two out of the three points $id-2$, $id$, and
  $id+2$, which we call the \emph{requests of pair} $i$. 
  We call $id$ the \emph{hub}, and we call $id-2$ and $id+2$ the \emph{fringe points} of pair $i$. 
  An optimal offline solution can easily serve all requests of any such instance at cost exactly $1$. 
  
  If any server of pair $i$ serves a request of another pair $j\neq i$, then it
  travels a distance of at least $d-3>2k$ in total, which
  implies a strict competitive ratio worse than $2k$.  We can thus assume
  without loss of generality that only servers of pair $i$ serve requests of pair $i$. 
  
  Any instance first presents two requests of pair 1, then two of pair 2, and
  so on.  Consider the two requests for any given pair that are about to be
  presented.  Let the current positions of the two servers of pair $i$ be
  $id+x$ and $id+y$.  (For $i=1$, we have $x=y=0$, but we do not need to
  consider this case separately.)
  
  Without loss of generality, let $|x|\le|y|$ and $y\ge 0$. 
  We now describe how the requests of pair $i$ presented by the instance are chosen. 
  For simplicity of this description, assume that $i=0$ by an appropriate shift
  of the coordinate system. 
 
  \begin{itemize}[left=0em]
    \item \textit{Case 1: $|x|\geq 1/2$.}
      The hub $0$ is requested first. The algorithm incurs a cost
      of at least $|x|$ serving this request.  Depending on the position $z$ of the
      server that did not serve the 
      last request, we request $-2$ if $z$ is positive and $2$ otherwise. 
      This will add a cost of at least $2$, leaving us with a total cost of 
      at least $2+|x|\geq 5/2$ in this case. 
  
    \item \textit{Case 2:  $|x|<1/2$.}
      $-2$ is requested first, leading to a cost of at least
      $3/2$. The server that did not serve this request has distance at least $1$ to
      at least one of the points $0$ or $2$. Such a point is requested, leading to a
      cost of at least $5/2$ for serving the requests of the considered pair.
  \end{itemize}
  
  The cost per server pair is thus at least $5/2$, yielding 
  a lower bound of $k/2\cdot 5/2=5k/4$ on the strict competitive ratio of any
  deterministic algorithm. 
\end{proof} 

\paragraph*{\boldmath Lower Bound of $2k-1$ on General Graphs.}
The most celebrated result regarding the $k$-server problem is the proof of
Koutsoupias and Papadimitriou \cite{KP1995a} showing the \algWFA to be 
$(2k-1)$-competitive in the distance model.  Intriguingly, $2k-1$ is also our
best \emph{lower} bound for the time model, stated in the following theorem. 

\begin{theorem}\label{thm:lowertwokminusonefinite}
  For any $k\geq 1$, there is a finite metric space on which no online algorithm
  for the time-optimal $k$-server problem has a competitive ratio better than $2k-1$. 
\end{theorem}

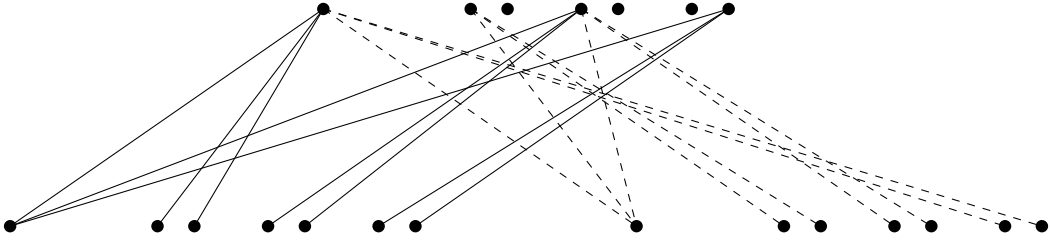
\begin{figure}[t]
    \begin{center}
      \begin{tikzpicture}[scale=8,yscale=.99,graphnode/.style={fill, circle, inner sep=1.65pt},firstchoice/.style={},secondchoice/.style={dashed}]
\def\endspace{0}
\def\blockspace{3}
\def\centerspace{2}
\def\optspace{0.5}
\node[graphnode] at (0.5151515151515151,0) {};
\node[graphnode] at (0.7575757575757576,0) {};
\node[graphnode] at (0.8181818181818181,0) {};
\node[graphnode] at (0.9393939393939394,0) {};
\node[graphnode] at (1.0,0) {};
\node[graphnode] at (1.121212121212121,0) {};
\node[graphnode] at (1.1818181818181817,0) {};
\node[graphnode] at (0.0,-0.36363636363636365) {};
\node[graphnode] at (0.24242424242424243,-0.36363636363636365) {};
\node[graphnode] at (0.30303030303030304,-0.36363636363636365) {};
\node[graphnode] at (0.42424242424242425,-0.36363636363636365) {};
\node[graphnode] at (0.48484848484848486,-0.36363636363636365) {};
\node[graphnode] at (0.6060606060606061,-0.36363636363636365) {};
\node[graphnode] at (0.6666666666666666,-0.36363636363636365) {};
\node[graphnode] at (1.0303030303030303,-0.36363636363636365) {};
\node[graphnode] at (1.2727272727272727,-0.36363636363636365) {};
\node[graphnode] at (1.3333333333333333,-0.36363636363636365) {};
\node[graphnode] at (1.4545454545454546,-0.36363636363636365) {};
\node[graphnode] at (1.5151515151515151,-0.36363636363636365) {};
\node[graphnode] at (1.6363636363636365,-0.36363636363636365) {};
\node[graphnode] at (1.696969696969697,-0.36363636363636365) {};
\draw[firstchoice] (0.5151515151515151,0)--(0.0,-0.36363636363636365);
\draw[secondchoice] (0.5151515151515151,0)--(1.0303030303030303,-0.36363636363636365);
\draw[firstchoice] (0.5151515151515151,0)--(0.24242424242424243,-0.36363636363636365);
\draw[firstchoice] (0.5151515151515151,0)--(0.30303030303030304,-0.36363636363636365);
\draw[firstchoice] (0.9393939393939394,0)--(0.42424242424242425,-0.36363636363636365);
\draw[firstchoice] (0.9393939393939394,0)--(0.48484848484848486,-0.36363636363636365);
\draw[firstchoice] (0.9393939393939394,0)--(0.0,-0.36363636363636365);
\draw[firstchoice] (1.1818181818181817,0)--(0.6060606060606061,-0.36363636363636365);
\draw[firstchoice] (1.1818181818181817,0)--(0.6666666666666666,-0.36363636363636365);
\draw[firstchoice] (1.1818181818181817,0)--(0.0,-0.36363636363636365);
\draw[secondchoice] (0.5151515151515151,0)--(1.6363636363636365,-0.36363636363636365);
\draw[secondchoice] (0.5151515151515151,0)--(1.696969696969697,-0.36363636363636365);
\draw[secondchoice] (0.7575757575757576,0)--(1.2727272727272727,-0.36363636363636365);
\draw[secondchoice] (0.7575757575757576,0)--(1.3333333333333333,-0.36363636363636365);
\draw[secondchoice] (0.7575757575757576,0)--(1.0303030303030303,-0.36363636363636365);
\draw[secondchoice] (0.9393939393939394,0)--(1.4545454545454546,-0.36363636363636365);
\draw[secondchoice] (0.9393939393939394,0)--(1.5151515151515151,-0.36363636363636365);
\draw[secondchoice] (0.9393939393939394,0)--(1.0303030303030303,-0.36363636363636365);
\end{tikzpicture} \caption{A small part of the construction used to
      prove \cref{thm:lowertwokminusonefinite} for $k=3$. One of the
      $k(k-1)^{k-1}=12$ blocks of one layer at the top (with the hub point on
      the left and the fringe points in $k$ groups of $k-1$ to the right), two
      blocks of the next layer at the bottom left and bottom right, and the
      edges induced by the corresponding choices of $k-1$ groups and one fringe
      point per chosen group in the block above; with dashed edges for the
      choice represented by the bottom right block.}
      \label{fig:lowerboundexample3}
    \end{center}
\end{figure}

\begin{proof}
  We first give a general idea of the proof. The metric space is a finite,
  unweighted graph. The vertex set is partitioned into three uniform layers, each
  of which contains many blocks of vertices. Each block in turn contains a
  \emph{hub} vertex and $k$ \emph{fringe groups}, each consisting of $k-1$
  \emph{fringe vertices}.  The instance is constructed in phases that go through
  the three layers cyclically. In each phase, $k$ points in one block of the
  current layer are requested. 
  The edges are chosen such that any phase can be interpreted as simulating an
  instance of the following problem: Consider a star with $k$ rays of length $2$. 
  The center corresponds to a \emph{hub} and the endpoints of the star to $k$
  fringe points. Let one server be positioned at each midpoint of a ray.  Now the
  instance requests the hub. The algorithm must move one server there at
  cost $1$. With one of the $k$ servers at the hub, one ray does not contain a
  server. Its fringe point is requested next, causing a cost of at least $2$. We
  do not allow servers that have served a request to move again during the phase
  by requesting previously requested points again, if necessary. Thus this
  process of requesting a fringe point of an empty ray can be repeated exactly
  $k-1$ times, causing a cost of $2k-1$ in the phase.

  The main challenge lies in simulating this instance (including the inability of
  servers to move after having served a request in any given phase) while making
  the process repeatable, despite having only a finite number of layers.  This
  requires the vertices in our graph to function simultaneously as either fringe
  point or hub in the layer where the requests appear and as the midpoint of a
  ray when chosen as a phase-initial server position.
  
  \medskip
  
  \noindent\textit{Description of the metric space.} The metric space is a
  tripartite graph, i.e., its vertices can be partitioned into three disjoint
  \emph{layers} $L_1\cup L_2\cup L_3=V$ such that the edges can be partitioned as
  $E_1\cup E_2\cup E_3=E$ with $E_1\subseteq L_1\times L_2$, $E_2\subseteq
  L_2\times L_3$, and $E_3\subseteq L_3\times L_1$.
  
  For all instances in the distribution we describe later, the three layers will
  be used cyclically by any optimal solution: its servers will all move
  synchronously first from $L_1$ to $L_2$, then to $L_3$, then back to $L_1$, and
  so on.  Moreover, the three subgraphs $(L_1\cup L_2,E_1)$, $(L_2\cup L_3,E_2)$,
  and $(L_3\cup L_1,E_3)$ are all isomorphic. 
  In particular, the layers all have the same size. Each layer is partitioned
  into $\Block\coloneqq k(k-1)^{k-1}$ uniformly sized \emph{blocks}. Each block
  in turn consists of one special vertex called its \emph{hub} and $k(k-1)$
  vertices that we call its \emph{fringe points},
  grouped into $k$ groups of size $k-1$. 
  We identify the blocks with the integers $[B]$ and similarly denote the groups
  of each block by $[k]$ and the points of a group by $[k-1]$. We denote the
  hub of block $\block$ in layer $L_\ell$ by $\cent(\ell,\block)$ and the
  fringe point $\vertex$ in
  group $\group$ of the same block by $\peri(\ell,\block,\group,\vertex)$. Thus
  the vertex set is given by
  \[ V\coloneqq \bigcup_{\ell=1}^{3}\bigcup_{\block=1}^{\Block}\left(\{\cent(\ell,\block)\}\cup \bigcup_{\group=1}^{k}\bigcup_{\vertex=1}^{k-1}\{\peri(\ell,\block,\group,\vertex)\}\right)\;. \]
  
  Note that the order of the graph is $\lvert V\rvert=3k(k-1)^{k-1}(1+k(k-1))\leq
  3k^3(k-1)^{k-1}\in \mathcal{O}(k^{k+2})$.
    
  For notational convenience we define $L_4\coloneqq L_1$ and $L_0\coloneqq L_3$.
  To describe the edges of $G$, we first highlight some uniform aspects of the construction. 
  For each $\ell\in[3]$, the edges $E_\ell$ between $L_\ell$ and $L_{\ell+1}$ are constructed identically. 
  
  Moreover, the edges from any layer to the next are \emph{block-uniform} in the following sense:  
  for any $\ell\in[3]$, $\group\in[k]$, and $\vertex\in[\Vertex]$, the
  neighborhood in $L_{\ell+1}$ is the same for every fringe point
  $\peri(\ell,\block,\group,\vertex)$, and the same for any hub $\cent(\ell,b)$
  independent of $\block\in[\Block]$.
  
  Now consider any block $\block$ of $L_\ell$. Recall that each block consists of
  a hub point and $k$ fringe groups, each with $k-1$ fringe points.  To each
  possible choice of $k-1$ fringe points from $k-1$ distinct groups, we assign one
  block $\choice\in[\Block]$ from $L_{\ell+1}$ by an arbitrary but fixed
  bijection. Such a bijection exists, because there are $k$ possibilities to
  choose $k-1$ fringe groups and $(k-1)^{k-1}$ ways to choose one fringe point
  from each of these groups, which means that the total number of possible
  choices is exactly $\Block$, the number of blocks in $L_{\ell+1}$.
  
  Fix such a choice and the corresponding block $\choice$ in $L_{\ell+1}$.  Each
  of the $k-1$ chosen fringe points, as well as the hub of $b$, is assigned
  injectively one of the $k$ groups of fringe points of $\choice$ and then made
  adjacent to all $k-1$ points of the assigned group.
  For concreteness, let the chosen fringe point of group $\group$ in $\block$ be
  assigned group $\group$ in block $\choice$, while the hub is assigned the
  remaining group of block $\choice$ that is not represented in the $k-1$ chosen
  ones.  Finally, all $k-1$ chosen fringe points and the hub of $\block$ are made
  adjacent to the hub of $\choice$ in $L_{\ell+1}$.
    
  More formally, identify the possible choices of $k-1$ fringe points in a block
  of $L_{\ell}$ with $[\Block]$.  For each choice $\choice\in[\Block]$, denote by
  $\group_\choice\in[k]$ the unique index of the group without a chosen point. 
  For any $\group\in[k]$ with $\group\neq \group_\choice$, denote by
  $\vertex_\choice^\group\in [k-1]$ the index of the unique point chosen from
  group $\group$. This means that the fringe points in choice $\choice$ are
  exactly $\peri(\ell,\block,\group,\vertex_\choice^\group)$ for $\group\in[k]$
  and $\group\neq \group_\choice$. Then we have the following edges between block
  $\block$, which is in $L_{\ell}$, and $L_{\ell+1}$:

  \begin{align*}
    E(\ell,\block)&\coloneqq \bigcup_{\choice=1}^{\Block} \{\{\cent(\ell,\block),\cent(\ell+1,\choice)\}\}\\
    &\,\cup \bigcup_{\choice=1}^{\Block}\, \bigcup_{\vertex=1}^{k-1} \{\{\cent(\ell,\block),\peri(\ell+1,\choice,\group_\choice,\vertex)\}\}\\
    &\,\cup \bigcup_{\choice=1}^{\Block}\bigcup_{\substack{\group\in[k]\\\group\neq \group_\choice}} \{\{\peri(\ell,\block,\group,\vertex_\choice^\group),\cent(\ell+1,\choice)\}\}\\
    &\,\cup \bigcup_{\choice=1}^{\Block}\bigcup_{\substack{\group\in[k]\\\group\neq \group_\choice}} \bigcup_{\vertex=1}^{k-1} \{\{\peri(\ell,\block,\group,\vertex_\choice^\group),\peri(\ell+1,\choice,\group,\vertex)\}\}\;.
  \end{align*}
  
  Note again that the vertices in $L_{\ell+1}$ do not depend on $\block$.
  The full set containing $3B^2k^2=3k^4(k-1)^{2k-2}$ edges is then 
  \[ E \coloneqq\bigcup_{\ell=1}^{3}\bigcup_{\block=1}^{\Block} E(\ell,\block)\;. \]
  
  \begin{figure}[t]
    \captionsetup[subfigure]{justification=centering}
    \begin{subfigure}{.495\textwidth}
    \centering
        \scalebox{0.5}{
\begin{tikzpicture}[scale=1,graphnode/.style={fill, circle, inner sep=0.75pt}]
\def\endspace{5}
\def\blockspace{3}
\def\centerspace{2}
\def\optspace{0.5}

\node[graphnode] at (3.1578947368421053,0) {};
\node[graphnode] at (4.421052631578947,0) {};
\node[graphnode] at (5.052631578947368,0) {};
\node[graphnode] at (6.947368421052632,0) {};
\node[graphnode] at (8.210526315789474,0) {};
\node[graphnode] at (8.842105263157894,0) {};
\node[graphnode] at (4.421052631578948,-7.657487780830827) {};
\node[graphnode] at (3.789473684210526,-6.56356095499785) {};
\node[graphnode] at (3.473684210526316,-6.016597542081363) {};
\node[graphnode] at (2.526315789473684,-4.3757073033319) {};
\node[graphnode] at (1.894736842105263,-3.281780477498925) {};
\node[graphnode] at (1.5789473684210529,-2.734817064582438) {};
\node[graphnode] at (10.421052631578949,-2.7348170645824377) {};
\node[graphnode] at (9.789473684210526,-3.8287438904154127) {};
\node[graphnode] at (9.473684210526315,-4.3757073033319) {};
\node[graphnode] at (8.526315789473685,-6.016597542081363) {};
\node[graphnode] at (7.894736842105264,-7.110524367914338) {};
\node[graphnode] at (7.578947368421052,-7.6574877808308255) {};
\draw[dashed] (0,0)--(12,0);
\draw[dashed] (0,0)--(6.0,-10.392304845413264);
\draw[dashed] (12,0)--(6.0,-10.392304845413264);
\draw[blue, opacity=1] (3.1578947368421053,0)--(4.421052631578948,-7.657487780830827);
\draw[blue, opacity=1] (3.1578947368421053,0)--(2.526315789473684,-4.3757073033319);
\draw[blue, opacity=1] (3.1578947368421053,0)--(3.473684210526316,-6.016597542081363);
\draw[blue, opacity=1] (4.421052631578947,0)--(3.789473684210526,-6.56356095499785);
\draw[blue, opacity=1] (4.421052631578947,0)--(4.421052631578948,-7.657487780830827);
\draw[blue, opacity=1] (3.1578947368421053,0)--(1.894736842105263,-3.281780477498925);
\draw[blue, opacity=1] (5.052631578947368,0)--(1.5789473684210529,-2.734817064582438);
\draw[blue, opacity=1] (5.052631578947368,0)--(2.526315789473684,-4.3757073033319);
\draw[blue, opacity=1] (6.947368421052632,0)--(4.421052631578948,-7.657487780830827);
\draw[blue, opacity=1] (6.947368421052632,0)--(2.526315789473684,-4.3757073033319);
\draw[blue, opacity=1] (6.947368421052632,0)--(3.473684210526316,-6.016597542081363);
\draw[blue, opacity=1] (8.210526315789474,0)--(3.789473684210526,-6.56356095499785);
\draw[blue, opacity=1] (8.210526315789474,0)--(4.421052631578948,-7.657487780830827);
\draw[blue, opacity=1] (6.947368421052632,0)--(1.894736842105263,-3.281780477498925);
\draw[blue, opacity=1] (8.842105263157894,0)--(1.5789473684210529,-2.734817064582438);
\draw[blue, opacity=1] (8.842105263157894,0)--(2.526315789473684,-4.3757073033319);
\draw[red, opacity=1] (4.421052631578948,-7.657487780830827)--(10.421052631578949,-2.7348170645824377);
\draw[red, opacity=1] (4.421052631578948,-7.657487780830827)--(8.526315789473685,-6.016597542081363);
\draw[red, opacity=1] (4.421052631578948,-7.657487780830827)--(9.473684210526315,-4.3757073033319);
\draw[red, opacity=1] (3.789473684210526,-6.56356095499785)--(9.789473684210526,-3.8287438904154127);
\draw[red, opacity=1] (3.789473684210526,-6.56356095499785)--(10.421052631578949,-2.7348170645824377);
\draw[red, opacity=1] (4.421052631578948,-7.657487780830827)--(7.894736842105264,-7.110524367914338);
\draw[red, opacity=1] (3.473684210526316,-6.016597542081363)--(7.578947368421052,-7.6574877808308255);
\draw[red, opacity=1] (3.473684210526316,-6.016597542081363)--(8.526315789473685,-6.016597542081363);
\draw[red, opacity=1] (2.526315789473684,-4.3757073033319)--(10.421052631578949,-2.7348170645824377);
\draw[red, opacity=1] (2.526315789473684,-4.3757073033319)--(8.526315789473685,-6.016597542081363);
\draw[red, opacity=1] (2.526315789473684,-4.3757073033319)--(9.473684210526315,-4.3757073033319);
\draw[red, opacity=1] (1.894736842105263,-3.281780477498925)--(9.789473684210526,-3.8287438904154127);
\draw[red, opacity=1] (1.894736842105263,-3.281780477498925)--(10.421052631578949,-2.7348170645824377);
\draw[red, opacity=1] (2.526315789473684,-4.3757073033319)--(7.894736842105264,-7.110524367914338);
\draw[red, opacity=1] (1.5789473684210529,-2.734817064582438)--(7.578947368421052,-7.6574877808308255);
\draw[red, opacity=1] (1.5789473684210529,-2.734817064582438)--(8.526315789473685,-6.016597542081363);
\draw[green, opacity=1] (10.421052631578949,-2.7348170645824377)--(3.1578947368421053,0);
\draw[green, opacity=1] (10.421052631578949,-2.7348170645824377)--(6.947368421052632,0);
\draw[green, opacity=1] (10.421052631578949,-2.7348170645824377)--(5.052631578947368,0);
\draw[green, opacity=1] (9.789473684210526,-3.8287438904154127)--(4.421052631578947,0);
\draw[green, opacity=1] (9.789473684210526,-3.8287438904154127)--(3.1578947368421053,0);
\draw[green, opacity=1] (10.421052631578949,-2.7348170645824377)--(8.210526315789474,0);
\draw[green, opacity=1] (9.473684210526315,-4.3757073033319)--(8.842105263157894,0);
\draw[green, opacity=1] (9.473684210526315,-4.3757073033319)--(6.947368421052632,0);
\draw[green, opacity=1] (8.526315789473685,-6.016597542081363)--(3.1578947368421053,0);
\draw[green, opacity=1] (8.526315789473685,-6.016597542081363)--(6.947368421052632,0);
\draw[green, opacity=1] (8.526315789473685,-6.016597542081363)--(5.052631578947368,0);
\draw[green, opacity=1] (7.894736842105264,-7.110524367914338)--(4.421052631578947,0);
\draw[green, opacity=1] (7.894736842105264,-7.110524367914338)--(3.1578947368421053,0);
\draw[green, opacity=1] (8.526315789473685,-6.016597542081363)--(8.210526315789474,0);
\draw[green, opacity=1] (7.578947368421052,-7.6574877808308255)--(8.842105263157894,0);
\draw[green, opacity=1] (7.578947368421052,-7.6574877808308255)--(6.947368421052632,0);
\end{tikzpicture}}
        \caption{Graph on $18$ vertices with $48$ edges for $k=2$.}
        \label{fig:lowerboundgraph1}
     \end{subfigure}
     \begin{subfigure}{.495\textwidth}
     \centering
     	  \include{triangle2.tex}
        \caption{Graph on $252$ vertices and $3888$ edges for $k=3$.}
        \label{fig:lowerboundgraph2}
     \end{subfigure}
     \caption{The graph describing the metric space used to prove
     \cref{thm:lowertwokminusonefinite} for $k=2$ and $k=3$. The layers $L_1$,
     $L_2$, and $L_3$ are arranged counterclockwise with $L_1$ at the top. Edges
     from $L_1$ to $L_2$ are shown in blue, those from $L_2$ to $L_3$ in red, and
     those from $L_3$ to $L_1$ in green.}
     \label{fig:lowerboundgraph}
  \end{figure}
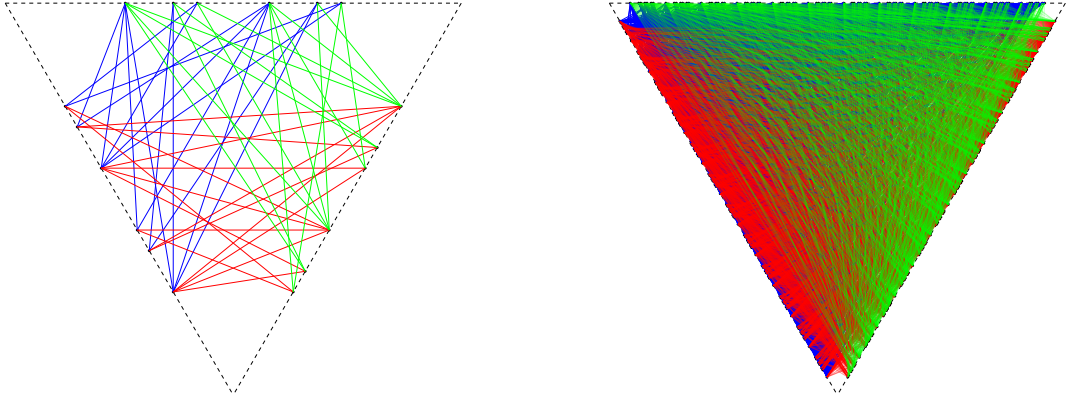
  
  A visual representation of the graph $G=(V,E)$ for $k=2$ is shown in
  \cref{fig:lowerboundgraph1}. For $k=3$, an example of one block in $L_1$ and
  two blocks in $L_2$ is given in \cref{fig:lowerboundexample3}, while the full
  graph is shown in \cref{fig:lowerboundgraph2}.
    
  The following claim summarizes some important properties of $G$. It follows
  directly from the construction of the edges. 

  \begin{claim}\label{claim:graphproperties}
    Fix any block $\block$ in $L_\ell$. Then we have the following properties:
    \begin{enumerate}[label=(\roman*)]
      \item\label{claim:blocks1} Any vertex in $L_{\ell-1}$ is adjacent to vertices of 
      at most one fringe group of block~$\block$.
      \item\label{claim:blocks2} Any fringe point in $L_{\ell+1}$ is adjacent to at most one fringe point of block~$\block$. 
      \item\label{claim:blocks3} Any hub point in $L_{\ell+1}$ is adjacent to 
      exactly one vertex for each of exactly $k-1$ of the $k$ fringe groups of block~$\block$.
    \end{enumerate}
  \end{claim}

  \noindent\textit{Description of the instances.} 
  We now show that no online algorithm can be better than $(2k-1)$-competitive on
  the metric space described by $G$. 
    
  Let \alg be any deterministic online algorithm. We will compare its cost against that of an offline algorithm \algOff on the   instance we describe now. The instance consists of distinct phases. In each
  phase, exactly $k$ distinct points, all from the same layer, are requested,
  potentially with many repetitions. The instances are designed such that
  \algOff can move all its $k$ servers by a distance of $1$ at the beginning of
  each phase and then serve all requests within this phase at no additional cost,
  while \alg has cost at least $2k-1$.
  
  At the start of each phase $\phase$ (and without loss of generality also in the
  initial configuration), the servers of \algOff are located in a single block
  $\block$ of $L_\ell$ and more specifically at one of its hub points and $k-1$
  points chosen from $k-1$ distinct fringe groups in it.
    
  The phase consists of vertices from the block $\choice$ in $L_{\ell+1}$
  corresponding to the choice of fringe points in block $\block$ that form the
  configuration of \algOff at the start of the phase.  Specifically, it consists
  of the hub of block $\choice$ and $k-1$ fringe points chosen from distinct
  groups.  
  This means that the process can thus be repeated since the final configuration
  is equivalent to the starting position.  It also means that \algOff can reach
  the configuration with cost $1$: it can move $k-1$ of its servers in block
  $\block$ to fringe points in block $\choice$ in the unique group they are
  adjacent to, and move the final server to the hub of block $\choice$.
    
  We can assume that within each phase, \alg always has a server on each
  previously requested point. Otherwise, the instance will simply request that
  point again, at no additional cost to \algOff.
    
  The instance starts the phase with a request to the hub of block $\choice$.
  \alg must move a server to that vertex and thus incurs cost $1$. Now assume
  that $\m\in[k-1]$ fringe points in block $\choice$ have already been
  requested by the instance in this phase. We claim that there is at least one
  fringe point $\vertex$ in block $c$ such that
  \begin{enumerate}[label=(\alph*)]   
    \item no vertex from the group in block $\choice$ containing $\vertex$ was requested in this phase before, and 
    \item no server of \alg can reach $\vertex$ with cost below $2$.
  \end{enumerate}

  By assumption, $\m+1$ servers of \alg already occupy requested points in block $\choice$. 
  These points are not adjacent to any other point in $L_2$, and thus in
  particular not to any point in block $\choice$.  There are $k-\m$ groups in
  block $\choice$ that are not yet represented by a request in this phase. 
  Since there are only $k-\m-1$ remaining servers, by
  \cref{claim:graphproperties}\cref{claim:blocks1,claim:blocks2}, there is
  at least one group in block $\choice$ such that no fringe point in this group
  is adjacent to a server that is in $L_{\ell-1}$ or on a fringe point in
  $L_{\ell+1}$. Within this group, there are $k-1$ points.  Any server of \alg
  that can reach such a point with cost at most $1$ is either on that point (in
  which case it cannot reach any other point in block $\choice$ with cost at
  most $1$) or on a hub in $L_{\ell+1}$ (in which case, by of
  \cref{claim:graphproperties}\cref{claim:blocks3}, the same holds). 
  
  Since there are $k-\m-1$ remaining servers, this means that such a point
  $\vertex$ exists, unless $\m=0$, i.e., if the only point requested in this
  phase has been the hub, and all servers of \alg occupy hubs in $L_{\ell+1}$
  or fringe points in $L_{\ell}$. In this case, by
  \cref{claim:graphproperties}\cref{claim:blocks3}, a server on a hub in $L_{\ell+1}$ can reach at
  most $k-1$ fringe points with cost $1$. A server on a fringe point in block
  $\choice$ can only reach one such point. Therefore, the total number of
  fringe points in block $\choice$ that can be reached in distance $1$ by
  servers on such vertices is at most $(k-1)^2$. This is strictly smaller than
  the total number of $k\cdot (k-1)$ fringe points in block $\choice$; so
  again, a point $\vertex$ with the required properties must exist.
  
  The instance now requests this point, which \alg must serve with cost at
  least $2$.  After $k-1$ such requests, the phase ends. \alg and \algOff are
  in the promised final configuration and the total cost of \alg in this phase
  is $1+(k-1)\cdot 2=2k-1$.
\end{proof}
 
\begin{observation}\label{obs:diameter}
  The diameter of the metric space constructed in the proof of
  \cref{thm:lowertwokminusonefinite} is $3$. 
\end{observation}

\begin{proof}
  To see this, observe that every hub point is adjacent to every hub point in
  the other layers by an edge of the form
  $\{\cent(\ell,\block),\cent(\ell+1,\choice)\}$.  Hence, hub points from
  different layers have distance $1$ from each other, while hub points of the
  same layer have distance $2$ from each other. 
  
  Each fringe point has the form $\peri(\ell,\block,\group,\vertex)$ and there is
  a $\choice\in[\Block]$ such that $\vertex=\vertex_\choice^\group$.  Thus every
  fringe point is adjacent to the hub point of some block in the next layer by an
  edge of the form 
  $\{\peri(\ell,\block,\group,\vertex_{\choice}^{\group}),\cent(\ell+1,\choice)\}$.
  Together with the first observation of this proof, we can conclude that fringe
  points from different layers have distance at most $3$ from each other, and
  that fringe points and hub points have distance at most $3$ from each other as well.
  
  It remains to show that the same holds for fringe points from the \emph{same}
  layer. 
  For this, we define $L_4\coloneqq L_1$, $L_5\coloneqq L_2$ and $L_6\coloneqq L_3$ for notational convenience.
  Consider two fringe points $\peri(\ell,b,g,n)$ and
  $\peri(\ell,b',g',n')$ for $1\leq \ell\leq 3$.  Choose any $c\in [\Block]$ such
  that $g_c\neq g'$ and $n'= n^{g'}_c$. This means that $f(\ell,b',g',n')$ is
  adjacent to the hub point $h(\ell+1,c)$.
  First assume that $g\neq g_{b}$. 
  Then choose any $c'\in[B]$ such that $g_{c'}=g$. The hub $h(\ell+1,c)$ is in
  turn adjacent to the fringe points $f(\ell+2,c',g_{c'},\tilde{n})=
  f(\ell+2,c',g,\tilde{n})$ for all $\tilde{n}\in [k-1]$. Lastly, the point
  $f(\ell+2,c',g,n^{g}_{b})$ is adjacent to $f(\ell+3,b,g,n)=f(\ell,b,g,n)$,
  since $g\neq g_{b}$. 
  Now assume that in fact $g=g_{b}$.
  Then $h(\ell+1,c)$ is adjacent to $h(\ell+2,c)$ (or any other hub in
  $L_{\ell+2}$), which in turn is adjacent to
  $f(\ell+3,b,g_b,\tilde{n})=f(\ell,b,g,\tilde{n})$ for all $\tilde{n}\in [k-1]$.
  
  To show that the diameter is no less than $3$, it suffices to find two
  distinct, nonadjacent points whose neighborhoods do not intersect.  This is the
  case for any two fringe points $f(\ell+1,c,g,n)$ and $f(\ell+1,c',g,n')$ with
  $n\neq n'$, $g=g_c$, and $g\neq g_{c'}$.
  They are distinct points in the same layer and thus not adjacent. Since
  $g=g_c$, $f(\ell+1,c,g,n)$ is not adjacent to any fringe points in $L_\ell$.
  Since $g\neq g_{c'}$, $f(\ell+1,c',g,n')$ is not adjacent to any hub points in
  $L_\ell$.  Additionally, since $n\neq n'$, for any choice $\tilde{c}\in [\Block]$,
  either $n^g_{\tilde{c}}\neq n$ or $n^{g}_{\tilde{c}}\neq n'$, so at least one of the two
  points will have no neighbors in block $\tilde{c}$ of layer $L_{\ell+2}$. Thus the
  full neighborhoods are indeed disjoint. 
       
  Note that this does not work in the special case of $k=2$, where $n\neq n'$ is
  not possible. In that case, it can be checked that $f(\ell+1,b,1,1)$ and
  $f(\ell+1,b,2,1)$ have disjoint neighborhoods: by
  \cref{claim:graphproperties}\cref{claim:blocks1}, their neighborhoods in $L_\ell$ are disjoint; and
  since for $k=2$ it holds for any choice $c\in [\Block]=[2]$ that either
  $g_c=1$ or $g_c=2$, their neighborhoods in $L_{\ell+2}$ are also disjoint.
\end{proof} 
  
It follows by \cref{corollary:uniformD} that there is a $3k$-competitive algorithm on this graph.

\subsection{Lower Bounds for Randomized Algorithms}\label{sec:lowerboundsrandomized}

In this section, we provide two lower bounds for randomized algorithms. 
Since it will be used in the following two proofs, we first state Yao's
principle for infinite minimization problems~{{\cite[Thm.~2.5]{Kom2016}}}. A
comprehensive explanation is found in the textbooks by Borodin and
El-Yaniv~\cite{BE1998} and Komm~\cite{Kom2016}.

\begin{fact}[Yao's Principle]\label{fact:yao}
  Let an online minimization problem be given. Let
  $\mathcal{I}_1,\mathcal{I}_2,\dots$ be an infinite sequence of sets of
  instances such that all instances in one set have the same length and such that
  the instances in $\mathcal{I}_{i+1}$ are longer than the instances in
  $\mathcal{I}_{i}$ for any positive integer $i$. Let $A_1,A_2,\dots$ be a list
  of all deterministic online algorithms. Let $\textup{Pr}_i$ denote an
  adversarial probability distribution on the instances in $\mathcal{I}_i$ and
  $\E_i$ the corresponding expected value function.  If there is some constant
  $c\ge 1$ such that 
  \begin{enumerate}[label=(\roman*)]
    \item\label{fact:yao1} $\displaystyle \frac{\min_i(\E_i[\cost{A_j(\mathcal{I}_i)}])}{\E_i[\cost{\Opt(\mathcal{I}_i)}]}\ge c$ for every positive integer $i$, and 
    \medskip
    \item\label{fact:yao2} $\displaystyle \lim\limits_{i\to\infty}\E_i[\cost{\Opt(\mathcal{I}_i)}]=\infty$,
  \end{enumerate}
  then there is, for any given positive constant $\eps<1$, no randomized online
  algorithm for the given problem that is $(c-\eps)$-competitive in expectation. 
\end{fact}

If the set of instances $\mathcal{I}$ is finite, we can formulate a statement analogous to \cref{fact:yao} but without
\cref{fact:yao2} that holds for the \emph{strict} competitive ratio. This is Yao's principle for finite minimization problems~{{\cite[Thm.~2.3]{Kom2016}}}. 

\paragraph*{\boldmath Lower Bound of $5k/6$ on the Line for the Strict Competitive Ratio.}
The first result is a lower bound on the strict expected competitive ratio on the line.

\begin{theorem}\label{thm:strictlowerlinerandomized}
  Let any positive even integer $k$ be given.  For any $\eps>0$, no randomized
  algorithm for the time-optimal $k$-sever problem on the real line has a strict
  expected competitive ratio better than $5k/6-\eps$.
\end{theorem}

\begin{proof}
  We provide a distribution of instances such that any deterministic algorithm
  cannot have an expected cost of less than $5k/6-\eps$ times the optimal cost.
  Applying Yao's principle for finite minimization
  problems~\cite[Thm.~2.3]{Kom2016} then yields the stated result. The
  distribution includes all hard instances used in the proof of
  \cref{thm:strictlowerline}, and more.
  
  Let $\eps>0$ and a deterministic algorithm \alg be given. Choose a positive
  $\delta<12\eps/k$ and $N\in \N$ such that $N\cdot \delta \geq 2$.  Let the
  servers be labeled $s_1,\dots,s_k$. 
  We group them into $k/2$ pairs; for each $i\in[k/2]$, we call $\{s_{2i-1},s_{2i}\}$ pair $i$. 
  For each pair, we position its two servers a distance of $2$ apart and 
  separate the pairs by large distances $d\ge 2k+4$ between all neighboring pairs. 
  Specifically, we can choose the initial positions of servers $s_{2i-1}$ and
  $s_{2i}$ to be $id-1$ and $id+1$ for each $i\in[k/2]$.  All instances request
  two or more points from the interval $[id-2,id+2]$ for
  each $i\in[k/2]$. We call the requests for points in $[id-2,id+2]$ the
  \emph{requests for pair} $i$.
  
  If any server of pair $i$ serves a request of another pair $j\neq i$, then it
  travels a distance of at least $d-3>2k$ during the entire instance, which
  implies a strict competitive ratio worse than $2k$.  Hence, we can assume
  without loss of generality that only servers of pair $i$ serve the requests of
  pair $i$. 
  
  \medskip
  \noindent\textit{Description of the distribution of instances.} 
  Any instance first presents the requests of pair $1$, 
  then the requests of pair $2$, and so on. 
  Consider the requests for any given pair $i$ that are about to be presented. 
  By shifting the coordinate system, we can assume that $i=0$, that is, we have
  requests for points in the interval $[-2,2]$ that will be served by one server
  pair. 
  
  With probability $1/2$, the first request is to the hub $0$.  Assuming this
  case, the second and last request of the server pair is to either $2$ or $-2$,
  each with probability $1/3$, and with the remaining probability of $1/3$ as
  follows: the second request is to $\delta$, the third one to $0$, and this is
  followed by another $N-2$ requests, always alternating between $\delta$ and
  $0$.
  
  With probability $1/2$, the first request is not to the hub, but instead
  to $2$ or $-2$, each with probability $1/4$.  The second request is then chosen
  uniformly at random among the two points in $\{-2,0,2\}$ not requested yet.
  
  \medskip
  \noindent\textit{Cost analysis.} 
  Let $x$ and $y$ be the positions of the two servers of the considered pair
  whose requests are presented next. As before, we may assume without loss of
  generality that the requests are points in the interval $[-2,2]$ and that
  $|x|\leq |y|$ and $y\geq 0$ by appropriately adjusting the coordinate system. 
  
  We consider the first request. If it is to $0$, \alg incurs cost at
  least $|x|$ serving it. If it is to $-2$, it incurs cost at least $2-|x|$. The
  expected cost for the first request is therefore at least $1/2\cdot
  |x|+1/4\cdot (2-|x|)=1/2+|x|/2\geq 1/2$.  Now we consider the second request.
  \alg's position may depend on the first request, so we consider the two cases
  separately. 
  
  \begin{itemize}[left=0em]
    \item \textit{Case 1: The first request was to $0$.}
    Then one server must occupy $0$. Let $z$ be the
    current position of the second server and assume for now that $z\geq 0$.  Then
    if the request is to $2$, \alg cannot serve it with cost less
    than $2-z$. If the request is to $-2$, then \alg cannot serve it with cost
    less than $2$. Finally, if the instance contains $N$ requests to alternately
    $0$ and $\delta$, then \alg could serve all of them with the server on $0$,
    which forces it to incur an additional cost of at least $N\cdot \delta\geq 2$.

    Alternatively, it can serve all requests to $0$ with the server already
    positioned there and all requests to $\delta$ with the second server, which
    lets it incur an additional cost of at least $z-\delta\leq 2$.  The case where
    $z<0$ is entirely symmetrical, except that \alg would
    have cost $|z|+\delta$ instead of $z-\delta$ in the final case.  Thus
    \alg's expected cost for serving the requests after the first one is at
    least $1/3\cdot (2-|z|+2+|z|-\delta)\geq 4/3-\delta/3$ in the considered event
    that $0$ is requested first. 
    
    \item \textit{Case 2: The first request was to $-2$ or $2$.}
    Then the second request is to one of the other two points among $\{-2,0,2\}$.
    The server that has served the first request has distance at least $2$ to
    both of these points; and the other
    server, wherever it is, has an average distance of at least $1$ to them. Since
    they are chosen with equal probability, \alg's expected cost for
    serving the second request is therefore at least $1$ if a point other than $0$
    is requested first. 
  \end{itemize}

  The algorithm thus incurs, for each pair, an expected cost of at least
  $1/2\cdot (4/3-\delta/3+1)=7/6-\delta/6$ serving all but the first request of
  the pair and thus $1/2+7/6-\delta/6=5/3-\delta/6$ including the first request.
  For the entire instance with all requests for the $k/2$ pairs, \alg's
  expected cost is therefore at least $k/2\cdot (5/3-\delta/6)=5/6\cdot
  k-\delta\cdot k/12\geq 5/6\cdot k-\eps$. Since \Opt can serve
  all requests at cost at most $1$ by moving the servers initially positioned at
  $id-1$ and $id+1$ to the appropriate two points among
  $\{id-2,id,id+\delta,id+2\}$ for each $i\in [k/2]$, the statement of the
  theorem follows.
\end{proof} 

\paragraph*{\boldmath Lower Bound of $k+H_k-1$ on General Graphs.}
In this last section, we provide a lower bound of $k+\Omega(\log k)$ on the
expected competitive ratio of randomized algorithms.  This is achieved on a
significant extension of the metric space used in the proof of
\cref{thm:lowertwokminusonefinite}. 
We describe a distribution of hard instances on which no deterministic
algorithm can defeat the lower bound to be proven, and finally apply Yao's
principle.  Note that the bound proven here for the time model 
is exponentially larger than the best known lower bound for the distance model,
which is polylogarithmic in~$k$.

\begin{theorem}\label{thm:lowertwokminusonerandomized}
  Let any positive integer $k\ge1$ be given.  No randomized online algorithm for
  the time-optimal $k$-server problem has a better constant competitive ratio
  than $k+H_k-1$. 
\end{theorem}

\begin{proof}
  First note that the statement is trivial for $k=1$ with $k+H_k-1=1$. 
  Now let an arbitrary integer $k\ge2$ and an arbitrarily small positive constant
  $\eps<1$ be given.  We construct a finite metric space and prove that no
  randomized algorithm for the time-optimal $k$-server problem is
  $(k+H_k-1-\eps)$-competitive on it. This will be done via Yao's principle as
  stated in \cref{fact:yao}, i.e., by proving a lower bound on deterministic
  algorithms on a hard distribution of instances instead. 
   
  Choose $\eps'>0$ such that $\eps'(k+H_k-1)\le \eps/2$; then choose $\delta>0$
  such that $(1-\delta)^2\ge1-\eps'$, which also implies $1-\delta\ge1-\eps'$;
  then choose $\Vertex\in \N$ such that both $\Vertex^2\geq k(k-1)$ and
  $k/\Vertex\leq \delta$; define $\Block \coloneqq k\Vertex^{k-1}$, and lastly choose
  $\phaselength\in\N$ such that $((k-1)/(k+1))^{\phaselength-1}\leq \delta$.  The
  metric space is described by a graph $G=(V,E)$ that extends the construction
  used in the proof of \cref{thm:lowertwokminusonefinite}.
  
  \medskip
  
  \noindent\textit{Description of the metric space.} 
  The construction of $G$ is similar to the graph used in the proof of
  \cref{thm:lowertwokminusonefinite}. It is also tripartite, consisting of
  three disjoint \emph{layers} $L_1\cup L_2\cup L_3=V$, with each layer being
  partitioned into $\Block$ uniformly sized \emph{blocks}. The main difference is
  that a block is further partitioned into $k+1$ groups,
  each consisting of a much larger number of exactly $\Vertex$ vertices. Among
  the $k+1$ groups of a block there is one distinguished \emph{hub group}, whose
  vertices are called \emph{hub points}, while the other $k$ groups are called
  \emph{fringe groups} and their vertices \emph{fringe points}. 
  For each layer, we identify its blocks with the natural numbers $[\Block]$, 
  for each block its fringe groups with $[k]$,
  and for each group its points with $[\Vertex]$.
  We denote the hub point $\vertex$ of block $\block$ of layer $\ell$ by
  $\cent(\ell,\block,\vertex)$ and the fringe point $\vertex$ of fringe group
  $\group$ of block $\block$ of layer $\ell$ by
  $\peri(\ell,\block,\group,\vertex)$.  Thus the vertex set $V$ has cardinality
  $\lvert V\rvert=3\Block(k+1)\Vertex=3k(k+1)\Vertex^k$ and is given by 
  \[V=\bigcup_{\ell=1}^{3}\bigcup_{\block=1}^{\Block}\left(\bigcup_{\vertex=1}^{\Vertex}\{\cent(\ell,\block,\vertex)\}\cup \bigcup_{\group=1}^{k}\bigcup_{\vertex=1}^{\Vertex}\{\peri(\ell,\block,\group,\vertex)\}\right)\,.\]
  
  For notational convenience we define $L_4\coloneqq L_1$ and $L_0\coloneqq L_3$. 
  To describe the edges of $G$, we first highlight some uniform aspects of the construction. 
  For each $\ell\in[3]$, the edges $E_\ell$ between $L_\ell$ and $L_{\ell+1}$ are constructed identically. 

  Moreover, the edges from any layer to the next are \emph{block-uniform} in the following sense:  
  For any $\ell\in[3]$, $\group\in[k]$, and $\vertex\in[\Vertex]$, the
  neighborhood in $L_{\ell+1}$ is the same for every fringe point
  $\peri(\ell,\block,\group,\vertex)$ independent of $\block\in[\Block]$.  For
  the hub points, there is not even a dependence on $\vertex$: for any fixed
  $\ell\in[3]$ and $\group\in[k]$, the neighborhood in $L_{\ell+1}$ of hub point
  $\cent(\ell,\block,\vertex)$ is the same for all choices $\block\in[\Block]$
  and $\vertex\in[\Vertex]$.
  
  Now consider any block $\block$ of $L_\ell$. Recall that each block consists of
  a hub group with $\Vertex$ hub points and $k$ fringe groups, each with
  $\Vertex$ fringe points.  To each possible choice of $k-1$ fringe points from
  $k-1$ distinct groups we assign one block $\choice\in[\Block]$ from
  $L_{\ell+1}$ by an arbitrary but fixed bijection. Such a bijection exists,
  because there are $k$ possibilities to choose $k-1$ fringe groups and
  $\Vertex^{k-1}$ ways to choose one fringe point from each of these groups,
  which means that the total number of possible choices is exactly $\Block$, the
  number of blocks in $L_{\ell+1}$.
  
  Consider now any such choice of $k-1$ fringe points in any block $\block$ of
  $L_\ell$ and the unique corresponding block $\choice$ in $L_{\ell+1}$ given by
  the fixed bijection. 
  We assign the fringe group $\group$ of block $\choice$ in $L_{\ell+1}$
  either, if there is a chosen point in group $\group$ of block $\block'$ of
  $L_\ell$, to this single point or, otherwise, to the entire hub group of
  $\choice$. Each chosen fringe point of $\choice$ is then adjacent to all
  $\Vertex$ points of the fringe group assigned to it. Moreover, all $\Vertex$
  points of the hub group of $\choice$ are adjacent to all $\Vertex$ points in
  the fringe group of $L_{\ell+1}$ assigned to this hub group.  Additionally,
  each of the chosen $k-1$ fringe points and all $\Vertex$ hub points of block
  $\block$ in $L_\ell$ are adjacent to all hub points of block $\choice$ in
  $L_{\ell+1}$. 
  
  More formally, identify the $\Block$ possibilities to choose $k-1$ fringe
  points in pairwise distinct fringe groups in a given block with $[\Block]$.
  This identification is used for all blocks uniformly.  For any choice
  $\choice\in [\Block]$, let $\group_{\choice}\in[k]$ denote the index of the
  unique fringe group in each block without a chosen point.
  Note that this index is block-uniform.
  Each other fringe group $\group\neq \group_{\choice}$ has a chosen point, which
  we denote by $\vertex_\choice^\group$,  which is again fully block-uniform. The
  chosen fringe points of choice $\choice$ in block $\block$ are thus exactly
  $\peri(\ell,\block,\group,\vertex_\choice^\group)$ for $\group\in[k]$ with
  $\group\neq \group_\choice$. 
  
  Formally, block $\block$ of layer $L_\ell$ has the following edges to next layer
  $L_{\ell+1}$:
  \begin{align*}
    E(\ell,\block)&\coloneqq \bigcup_{\choice=1}^{\Block}\bigcup_{\vertex'=1}^{\Vertex} \bigcup_{\vertex=1}^{\Vertex}\{\{\cent(\ell,\block,\vertex'),\cent(\ell+1,\choice,\vertex)\}\}\\
    &\qquad\cup \bigcup_{\choice=1}^{\Block} \bigcup_{\vertex'=1}^{\Vertex} \bigcup_{\vertex=1}^{\Vertex} \{\{\cent(\ell,\block,\vertex'),\peri(\ell+1,\choice,\group_{\choice},\vertex)\}\}\\
    &\qquad\cup \bigcup_{\choice=1}^{\Block} \bigcup_{\substack{g\in[k]\\\group\neq \group_{\choice}}} \bigcup_{\vertex=1}^{\Vertex}\{\{\peri(\ell,\block,\group,\vertex_{\choice}^{\group}),\cent(\ell+1,\choice,\vertex)\}\}\\
    &\qquad\cup\bigcup_{\choice=1}^{\Block} \bigcup_{\substack{\group\in[k]\\\group\neq \group_{\choice}}} \bigcup_{\vertex=1}^{\Vertex} \{\{\peri(\ell,\block,\group,\vertex_\choice^\group),\peri(\ell+1,\choice,\group,\vertex)\}\}\;.
  \end{align*}
  
  The entire edge set is thus 
  $E=\bigcup_{\ell=1}^{3}\bigcup_{\block=1}^{\Block} E(\ell,\block)$.
  The following claim is the analogue to 
  \cref{claim:graphproperties} in the proof of \cref{thm:lowertwokminusonefinite}
  and verified by direct inspection of $G$'s construction.
  
  \begin{claim}\label{claim:graphpropertiesrandom}
    Fix a block $\block$ in $L_{\ell}$. Then we have the following properties. 
    \begin{enumerate}[label=(\roman*)]
      \item\label{claim:blocks1random} Any point in $L_{\ell-1}$ is adjacent to the vertices of 
      at most one fringe group of block~$\block$.
      \item\label{claim:blocks2random} Any fringe point in $L_{\ell+1}$ is adjacent to 
      at most one fringe point of block~$\block$.
      \item\label{claim:blocks3random} Any hub point in $L_{\ell+1}$ is adjacent to 
        exactly one vertex for each of exactly $k-1$ of the $k$ fringe groups of block~$\block$.
    \end{enumerate}
  \end{claim}
  
  \noindent\textit{Description of the distribution of instances.} 
  We describe a distribution over instances such that the expected cost incurred
  by any deterministic online algorithm on this distribution surpasses 
  $k+H_k-1-\eps$ times the expected cost of an optimal offline algorithm.
  
  Let \alg be any deterministic online algorithm. We will compare its expected
  cost against that of an offline algorithm \algOff on the distribution of
  instances we describe now. Any instance consists of distinct phases. In each
  phase, exactly $k$ distinct points, all from the same layer, are requested,
  potentially with many repetitions. The instances are designed such that
  \algOff can move all its $k$ servers by a distance of $1$ at the beginning of
  each phase and then serve all requests within this phase at no additional cost,
  while \alg has an expected cost larger than 
  $k+H_k-1-\eps$.
  
  At the start of each phase $\phase$ (and without loss of generality also in the
  initial configuration), the servers of \algOff are located in a single block
  $\block$ of $L_\ell$, and more specifically at one of its hub points and $k-1$
  points chosen from the $k$ distinct fringe groups in it. 
  
  Any phase consists of requests for points from the unique block $\choice$ in
  $L_{\ell+1}$ corresponding to the choice of fringe points that form the
  configuration of \algOff at the start of the phase.  The first request of the
  phase is any hub point of $\block$, chosen uniformly at random. The entire
  phase $\phase$ is divided into $k$ subphases $\phase_0,\dots,\phase_{k-1}$.
  Subphase $\phase_0$ consists of only the first request. 
  For any $\m\in[k-1]$, subphase $\phase_{\m-1}$ ends as soon as the $\m$-th
  distinct point in the current phase has been requested and served, and
  $\phase_\m$ starts immediately after this. 
  
  In particular, from the start of subphase $\phase_\m$ until its last request
  appears, exactly $\m$ distinct points have been requested already in the phase. 
  We call a group \emph{unused} if no point from this group has been requested in the phase so far. 
  When subphase $\phase_\m$ has not ended yet, the next request is chosen as follows. 
  
  If there have been $\phaselength-1$ requests in the current subphase already,
  any point from any unused fringe group in $\block$ is requested uniformly at
  random. Note that this new request ends the subphase and allows us to bound the
  number of requests in any given subphase by~$\phaselength$.
  
  If there have been fewer requests in the subphase, then one group of block
  $\block$ (i.e., the hub group or any fringe group) is chosen uniformly at
  random. If the chosen group is unused, a point from it is chosen uniformly at
  random and requested, which also ends the subphase. Otherwise, the point of
  this group already requested during this phase is requested again.
  
  \medskip
  
  \noindent\textit{Offline cost per phase.} 
  During any given phase exactly one hub point of $\block$ and $k-1$ fringe
  points from distinct groups in $\block$ are requested, perhaps repeatedly.
  Therefore, \algOff can serve all requests of this phase at cost $1$: It
  simultaneously moves $k-1$ of the servers from their positions in $\block'$ to
  fringe points in $\block$ in the unique group they are adjacent to, and the
  remaining server to the requested hub of $\block$. This also implies that
  \algOff ends the phase in a configuration isomorphic to the initial
  configuration; thus the process can be repeated indefinitely. Note also that
  the length of a phase is at most $1+\sum_{\m=1}^{k-1}
  \phaselength=1+(k-1)\phaselength$.
  
  \medskip
  
  \noindent\textit{Online cost per phase.} 
  We will show that in each phase, the algorithm has expected cost more than
  $k+H_k-1-\eps$.  First, \alg can only serve the hub point requested at the
  beginning for free if a server already covers it; otherwise, it incurs cost of
  at least $1$ to serve it.    
  Since the requested hub point in the block is chosen uniformly at random from
  the block's hub group of size $\Vertex$, the probability that it is occupied by
  one of the $k$ servers of \alg is at most $k/\Vertex$. Thus, the expected
  cost of \alg to serve the first request is at least $1-k/\Vertex\geq
  1-\delta$.
  
  Consider now the expected cost of \alg in subphase $\phase_\m$ for any
  $\m\in[k-1]$. During $\phase_\m$, exactly $\m$ distinct points have been
  requested already in the current phase. 
  
  We consider first the case that among these $\m$ points previously requested in
  the current phase, there is at least one point not occupied by \alg.  In this
  case there is a probability of $1/(k+1)$ for this unoccupied point to be
  requested next and \alg incurring a cost of at least $1$ serving it. 
  Moreover, since exactly $\m$ points have already been requested in the current
  phase during $\phase_\m$, there are still $k+1-\m$ unused groups.
  There is a probability of $({k+1-\m})/({k+1})$ that one of these unused groups
  is chosen for the next request, and in such a case the point from the chosen
  group is then chosen uniformly at random. Since at most $k$ of the points in
  the chosen group can be occupied, the probability that an unoccupied point is
  chosen is thus at least $1-k/\Vertex$.  Therefore, \alg incurs for every
  request in subphase $\phase_\m$ with $\m\in[k-1]$ an expected cost of at least 
  \begin{equation}\label{eq:expcost1}
    \frac{1}{k+1}+\frac{k+1-\m}{k+1}\cdot \left(1-\frac{k}{\Vertex}\right)\,.
  \end{equation}
  
  We now consider the case that the $\m$ servers of \alg do in fact occupy all
  previously requested points, meaning that they can serve any request to such a
  point with cost $0$. However, a new point is requested with probability
  \begin{equation}\label{eq:probnewpoint}
    \frac{k+1-\m}{k+1}\;.
  \end{equation}

  Given this event of a new point being requested, we consider the probability
  that a server of \alg can serve the current request with low cost. Such a
  server can serve the request with cost $0$ only if it is already occupying the
  requested point. Since this point is chosen uniformly at random among the
  $\Vertex$ points in its group, the probability of this occurring is at most
  $k/\Vertex$. Thus the probability that it has cost at least $1$ is at least
  \begin{equation}\label{eq:probcostoneplus}
    1-k/\Vertex\;.
  \end{equation}
  
  Next we derive an upper bound of $1/k\le 1/(k+1-\m)$ on the probability
  that a given server of \alg can serve a requested point from a block $\block$
  in layer $L_\ell$ with cost at most $1$. This will then yield a lower bound on
  the probability that the cost is at least $2$. We consider three cases
  according to the current layer of the server:
  \begin{itemize}[left=0em]
    \item \textit{Case 1.} If the considered server occupies a point in $L_{\ell-1}$, then it can
    serve the next request with cost at most $1$ only if this occupied point is
    adjacent to requested point of $\block$. There are $k+1-\m$ groups in $\block$
    from which one is chosen uniformly at random to contain the request. By
    \cref{claim:graphpropertiesrandom}\cref{claim:blocks1random}, the
    probability for the considered server to be adjacent to a point of the right
    group is therefore at most $1/(k+1-\m)$.
    \item \textit{Case 2.} If the considered server
    occupies a point in $L_{\ell}$, then it is trivially not adjacent to any other
    point in $L_{\ell}$. Thus it can serve the request with cost at most $1$ (and in
    fact $0$) if and only if it is exactly at the requested location already.
    Since in any group any of the $\Vertex$ points is chosen uniformly at random,
    this happens with a probability of at most $1/\Vertex$. Due to $\Vertex\ge k$,
    this is bounded by $1/k\le 1/(k+1-\m)$. 
    \item \textit{Case 3.} Assume now that the considered server occupies a point in $L_{\ell+1}$. 
    \begin{itemize}
    \item \textit{Case 3.1.} If the occupied point is a hub point, then it is adjacent to only $k-1$ fringe points in $\block$
    by \cref{claim:graphpropertiesrandom}\cref{claim:blocks3random}.  Since
    the requested point is chosen uniformly at random among $k-\m$ groups, each
    with $\Vertex$ points, the probability that the considered server can serve
    this request with cost at most $1$ is at most $(k-1)/(\Vertex\cdot (k-\m))\leq
    (k-1)/\Vertex$. 
    By our choice of $\Vertex$, this is bounded by $(k-1)/\Vertex\leq 1/k\leq 1/(k+1-\m)$.  
    \item \textit{Case 3.2.} If the occupied point is a fringe point, then it is, by
    \cref{claim:graphpropertiesrandom}\cref{claim:blocks2random},  adjacent to
    only a single fringe point in block $\block$ of layer $L_\ell$.  Thus, the
    probability of it serving the request with cost at most $1$ is at most
    $1/\Vertex\le 1/k\leq 1/(k+1-\m)$.
    \end{itemize} 
  \end{itemize}
  
  Thus, the probability that \emph{any} of the $k-\m$ servers of \alg not
  located at points previously requested during the current phase can serve the
  next request with cost $1$ or less is at most $(k-\m)/(k+1-\m)$ by the union
  bound.  The probability that \alg incurs cost at least $2$ is therefore at
  least
  \begin{equation}\label{eq:probcosttwoplus}
    1-\frac{k-\m}{k+1-\m}\ge \frac1{k+1-\m}\;.
  \end{equation}
  
  Using the tail-sum formula $\E[X]=\sum_{i=1}^{\infty} \P[X\geq i]\geq
  \P[X\geq 1]+\P[X\geq 2]$, the expected cost of \alg to serve the request
  given that $\m$ of its servers occupy the previously requested points is at
  least
  \begin{equation}\label{eq:expcost2}
    \frac{k+1-\m}{k+1}\cdot \left(1-k/\Vertex+\frac{1}{k+1-\m}\right)=\frac{1}{k+1}+(1-k/\Vertex)\cdot \frac{k+1-\m}{k+1}\;.
  \end{equation}
  
  The leading factor is the probability that the requested point has not yet been
  requested in the current phase; see \cref{eq:probnewpoint}. The probabilities
  that \alg has cost at least $1$ and $2$ are due to
  \cref{eq:probcostoneplus,eq:probcosttwoplus}, respectively.
  
  By \cref{eq:expcost1,eq:expcost2}, no matter where the servers
  of \alg are located, the expected cost to serve any single request in the
  subphase is at least 
  \begin{equation}\label{eq:expcostrequest}
    \frac{1}{k+1}+(1-k/\Vertex)\cdot \frac{k+1-\m}{k+1}\geq \frac{1}{k+1}+(1-\delta)\cdot \frac{k+1-\m}{k+1} \,.
  \end{equation}
  
  We now consider the expected cost of \alg over the entire subphase.
  For each $i\in[\phaselength]$, we define the random variable $Y_i$ that is $1$
  if there are at least $i$ requests in $\phase_\m$, and $0$ otherwise.
  
  In any subphase $\phase_\m$ with $s\in[k-1]$, the probability for requesting a
  point outside of the $\m$ points already requested in the considered phase and
  thus ending the subphase, is exactly $(k+1-\m)/(k+1)$. Since the subphase
  contains at least $i$ requests exactly if the first $i-1$ requests are points
  already requested, we know that
  $\P[Y_i=1]=(1-(k+1-\m)/(k+1))^{i-1}=(\m/(k+1))^{i-1}$ for any $i\in[
  \phaselength ]$.  We also define for every $i\in[\phaselength]$ the random
  variable $X_i$ to be $0$ if $Y_i=0$, i.e., if the request sequence of subphase
  $\phase_\m$ has fewer than $i$ requests, and as the cost of \alg to serve the
  $i$-th request otherwise. 
  
  We have proven with \cref{eq:expcostrequest} that $\E[X_i\mid Y_i=1]\geq
  1/(k+1)+(1-\delta)\cdot (k+1-\m)/(k+1)$ for $i\in[\phaselength-1]$. The cost of
  \alg in the entire subphase $\phase_\m$ is $\sum_{i=1}^{\phaselength}
  X_i\cdot Y_i$. And thus, the expected cost of \alg in $\phase_\m$ is
  \begin{align*}
  \E\left[\sum_{i=1}^{\phaselength} X_i\cdot Y_i\right]& =\sum_{i=1}^{\phaselength}\E[ X_i\cdot Y_i]\\
  &=\sum_{i=1}^{\phaselength}\E[ X_i\mid Y_i=1]\cdot \P[Y_i=1] \geq \sum_{i=1}^{\phaselength-1}\E[ X_i\mid Y_i=1]\cdot \P[Y_i=1]\\
  &\geq  \left(\frac{1}{k+1}+(1-\delta)\cdot \frac{k+1-\m}{k+1}\right)\cdot \sum_{i=0}^{\phaselength-2} \P[Y_{i+1}=1]\\
  &= \left(\frac{1}{k+1}+(1-\delta)\cdot\frac{k+1-\m}{k+1}\right)\cdot \sum_{i=0}^{\phaselength-2} \left(\frac{\m}{k+1}\right)^{i}\\
  &= \left(\frac{1}{k+1}+(1-\delta)\cdot\frac{k+1-\m}{k+1}\right)\cdot\frac{1-\left(\frac{\m}{k+1}\right)^{\phaselength-1}}{1-\frac{\m}{k+1}}\\
  &\geq \left(\frac{1}{k+1}+(1-\delta)\cdot\frac{k+1-\m}{k+1}\right)\cdot (1-\delta)\cdot \frac{k+1}{k+1-\m}\\
  &\geq  (1-\delta)^2\cdot \left(\frac{1}{k+1}+\frac{k+1-\m}{k+1}\right)\cdot \frac{k+1}{k+1-\m}\\
  &\geq (1-\eps')\cdot \left(1+\frac{1}{k+1-\m}\right)\,.
  \end{align*}
  
  Combined with the expected cost for the first request, the expected cost for each phase is thus at least
  \begin{align*}
  (1-\delta)+\sum_{\m=1}^{k-1}(1-\eps')\cdot \left(1+\frac{1}{k+1-\m}\right)
  & \geq(1-\eps')\cdot \left(1+(k-1)+\sum_{\m=1}^{k-1}\frac{1}{k+1-\m}\right)\\
  & =(1-\eps')\cdot (k+H_k-1)\geq k+H_k-1-\eps/2. 
  \end{align*}
  
  \medskip
  
  \noindent\textit{Applying Yao's principle.} 
  To apply Yao's principle (more precisely \cref{fact:yao}, and thus that
  \cref{fact:yao2} holds), it now suffices to show that for any constant $\ell$
  there is an instance of the described type whose optimal cost is at least $\ell$. This
  is easily seen: Note first that any phase consists of at most
  $1+\sum_{\m=1}^{k-1} \phaselength$ requests, thus our construction yields
  instances with arbitrarily many phases. Any two consecutive phases contain
  together requests for $2k$ distinct points, $k$ in one layer and $k$ in another
  one. Thus any solution incurs a cost of at least $1$ during these two phases.
  We can construct for any given $\ell \in\N$ an instance that contains $2\ell$
  phases and thus, by a very rough estimation, lets any algorithm incur a cost of
  at least $\ell$.
  
  Finally, applying Yao's principle thus proves that, for any $\eps''>0$, there
  is no randomized online algorithm for the time-optimal $k$-server problem with
  an expected competitive ratio of $k+H_k-1-\eps/2-\eps''$. Choosing
  $\eps''\coloneqq\eps/2$ yields the desired statement.
\end{proof}

Note that it can be shown that the graph used in the proof of
\Cref{thm:lowertwokminusonerandomized} has diameter $3$, analogously to
\Cref{obs:diameter}. Thus there is also a deterministic algorithm with
competitive ratio $3k$ on that metric space.

\section{Conclusion}\label{sec:conclusion}

We hope to initialize a line of research that focuses on the time-optimal
$k$-server problem, which has been neglected in favor of its classical
distance-optimal twin.  We have proven a series of lower bounds, showing the
time model to be harder on various metric spaces, including simple cycles and
all Euclidean spaces, 
which implies that the direct analogue of the $k$-server conjecture in the time
model cannot be true.  Our strongest lower bound matches---intriguingly---the
best upper bound known for the classical distance-optimal $k$-server problem,
which is attained by the deterministic work-function algorithm \algWFA. 
A priori, it could thus be true that this algorithm is in fact exactly
$(2k-1)$-competitive on general metric spaces in both models. 

A natural next step is to find good performance guarantees for the \algWFA in
the time-optimal setting.  Unfortunately, the celebrated analysis of
Koutsoupias and Papadimitriou~\cite{KP1995a} that proved so successful for the
distance model resists any straightforward adaption; multiple key concepts,
such as the duality between the so-called minimizers and maximizers, do not
translate well to the time model. Better upper bounds in the time model will
thus probably provide us with several novel techniques. 

Our preliminary attempts and experimentally gathered evidence indeed point to a
subquadratic upper bound.  We believe the lower bound of
\cref{thm:lowertwokminusonefinite} to be tight, i.e., we conjecture that $2k-1$
is indeed the best competitive ratio possible for a deterministic algorithm on
general metric spaces in the time mode.

\begin{conjecture}\label{conj:determ-time}
  There is a deterministic algorithm for the time-optimal \textit{k}-server
  problem with a competitive ratio of $2k-1$ on general metric spaces.
\end{conjecture}

For randomized algorithms, we analogously suspect an expected competitive ratio
of $c$ in the distance model to correspond to one of $c+k-1$ in the time model.

\section*{Acknowledgments}

We thank Peter Rossmanith for his suggestion to study the $k$-server problem
with time-based cost and Richard Královič for his idea for an algorithm on
metric spaces of bounded diameter.

\small

\bibliography{refs}
\bibliographystyle{plainurl}

\end{document}